\documentclass[12pt,reqno]{amsart}

\usepackage{graphicx}
\usepackage{hyperref} %%-This make a table of content in the pdf file. 
\hypersetup{
    colorlinks,
    citecolor=black,
    filecolor=black,
    linkcolor=black,
    urlcolor=black
}

\usepackage{color}

\newcommand{\Sch}{Schr\"odinger }
\newcommand{\szego}{Szeg\H{o} }

\newcommand{\baa}{\begin{align*}}
\newcommand{\eaa}{\end{align*}}
\newcommand{\bea}{\begin{eqnarray*} }
\newcommand{\eea}{\end{eqnarray*} }
\newcommand{\beq}{\begin{equation} }
\newcommand{\eeq}{\end{equation} }
\newcommand{\bp}{\begin{prop}}
\newcommand{\ep}{\end{prop}}
\newcommand{\bt}{\begin{theorem}}
\newcommand{\et}{\end{theorem}}
\newcommand{\bpf}{\begin{proof}}
\newcommand{\epf}{\end{proof}}
\newcommand{\bl}{\begin{lem}}
\newcommand{\el}{\end{lem}}
\newcommand{\bc}{\begin{cor}}
\newcommand{\ec}{\end{cor}}
\newcommand{\bd}{\begin{defn}}
\newcommand{\ed}{\end{defn}}
\newcommand{\be}{\begin{equation} }
\newcommand{\ee}{\end{equation} }
\newcommand{\bee}{\begin{eqnarray} }
\newcommand{\eee}{\end{eqnarray} }

\newcommand{\R}{{\mathbb R}}
\newcommand{\C}{{\mathbb C}}

\newcommand{\Z}{{\mathbb Z}}

\newcommand{\E}{{\mathbf E}}

\newcommand{\acal}{\mathcal{A}}

\newcommand{\ccal}{\mathcal{C}}
\newcommand{\dcal}{\mathcal{D}}

\newcommand{\fcal}{\mathcal{F}}

\newcommand{\hcal}{\mathcal{H}}

\newcommand{\ncal}{\mathcal{N}}

\newcommand{\rcal}{\mathcal{R}}
\newcommand{\scal}{\mathcal{S}}
\newcommand{\tcal}{\mathcal{T}}

\DeclareMathOperator{\Ai}{Ai}
\DeclareMathOperator{\Cov}{Cov}

\newcommand{\lan}{\langle}
\newcommand{\ran}{\rangle}

\newcommand{\pa}{\partial}
\newcommand{\set}[1]{\{#1\}}
\newcommand{\ihbar}{{\frac{i}{h}}}

\newcommand{\ZN}{|Z_{\Phi_{\hbar, E}}|}
\newcommand{\inprod}[2]{\ensuremath{\left\langle#1,#2\right\rangle}}
\newcommand{\twiddle}[1]{\ensuremath{\widetilde{#1}}}
\newcommand{\W}{\ensuremath{\Omega}}
\newcommand{\RM}{\backslash}
\newcommand{\gives}{\ensuremath{\rightarrow}}
\newcommand{\x}{\ensuremath{\times}}
\newcommand{\abs}[1]{\left\lvert #1 \right\rvert}

\newcommand{\lr}[1]{\ensuremath{\left(#1\right)}}
\newcommand{\dell}{\ensuremath{\partial}}
\newcommand{\half}{{\textstyle \frac 12}}
\newcommand{\vol}{{\operatorname{Vol}}}

\renewcommand{\Re}{{\text Re}}
\newcommand{\w}{\omega}

\newtheorem{theo}{{\sc Theorem}}[section]
\newtheorem{lem}[theo]{{\sc Lemma}}
\newtheorem{prop}[theo]{{\sc Proposition}}
\newtheorem{defn}[theo]{{\sc Definition}}
\newtheorem{cor}[theo]{{\sc Corollary}}
\newtheorem{remark}{{\sc Remark}}

\title[Harmonic Oscillator near the caustic]{Scaling of Harmonic Oscillator eigenfunctions and their nodal sets around the caustic }

\begin{document}
\author{Boris Hanin, Steve Zelditch, Peng Zhou}

\address{Department of Mathematics, Northwestern  University, Evanston, IL
60208, USA}
\email[S. Zelditch]{zelditch@math.northwestern.edu}
\email[P. Zhou]{pengzhou@math.northwestern.edu}
\address{Department of Mathematics, MIT, Cambridge, MA 02139}
\email[B. Hanin]{bhanin@mit.edu}

\thanks{SZ is partially supported by NSF grant DMS- 1541126, and BH is partially supported by NSF grant DMS-1400822.}
\maketitle

\begin{abstract}
We study the scaling asymptotics of the eigenspace projection  kernels $\Pi_{\hbar, E}(x,y)$ of the isotropic Harmonic Oscillator $\hat{H}_{\hbar} = - \hbar^2 \Delta +\abs{x}^2$ of eigenvalue   $E = \hbar(N + \frac{d}{2})$  in the semi-classical limit $\hbar \to 0$. The principal result is an explicit formula  for the scaling asymptotics of $\Pi_{\hbar, E}(x,y)$ for $x,y$ in a $\hbar^{2/3}$ neighborhood of
the caustic $\ccal_E$ as $\hbar \gives 0.$ The scaling asymptotics are applied to the distribution of nodal sets of Gaussian random eigenfunctions
  around the caustic
as $\hbar \to 0$.  In previous work we proved that the density of zeros of Gaussian random eigenfunctions of $\hat{H}_{\hbar}$ have different orders in the Planck constant $\hbar$ in the allowed and forbidden regions: In the allowed region the density is of order $\hbar^{-1}$ while it is $\hbar^{-1/2}$ in the forbidden region. 
%We use the scaling limit of $\Pi_{\hbar, E}$ to study the average density in %the transition region around the caustic between the allowed and forbidden %regions.
 Our main result on nodal sets is that the density of zeros is of order $\hbar^{-\frac{2}{3}}$ in an $\hbar^{\frac{2}{3}}$-tube around the caustic. This tube radius is the `critical radius'. For annuli of larger inner and outer radii $\hbar^{\alpha}$ with $0< \alpha < \frac{2}{3}$ we obtain density results which interpolate between this critical radius result and our prior ones in the allowed and forbidden region. We also show that the Hausdorff $(d-2)$-dimensional measure of the intersection of the nodal set with the caustic is of order $\hbar^{- \frac{2}{3}}$. 
\end{abstract}
%\tableofcontents
%%\setcounter{section}{-1}

\date{\today}

\section{Introduction}
This article is concerned with the scaling asymptotics of eigenspace projections of the isotropic Harmonic Oscillator
\begin{equation} \label{Hh}
\widehat{H}_{\hbar} =  \sum_{j = 1}^d \left(- \frac{\hbar^2}{2}   \frac{\partial^2 }{\partial
x_j^2} + \frac{x_j^2}{2} \right),
\end{equation}
and their applications to nodal sets of random Hermite eigenfunctions when $d\geq 2$. It is well-known that the spectrum of $\widehat{H}_{\hbar}$ consists of the eigenvalues 
\[E=\hbar\lr{N+d/2},\qquad N\in \Z_{\geq 0}.\]
The semi-classical limit at the energy level $E>0$ is the limit
as $\hbar \to 0, N \to \infty$ with fixed $E$, so that $\hbar$ only takes the values \[\hbar_N:=\frac{E}{N+d/2}.\]
We denote the corresponding eigenspaces by
\begin{equation} \label{VNE} V_{\hbar_N, E}: = \{\psi \in L^2(\R^d): \hat{H}_{\hbar_N} \psi = 
E \psi \}. \end{equation}
 The eigenspace projections are
the orthogonal projections
\begin{equation} \label{PiDEF} \Pi_{\hbar_N, E}: L^2(\R^d) \to V_{\hbar_N,E}. \end{equation}

An important feature of eigenfunctions of \Sch operators
$-\hbar^2 \Delta + V$ is that
as $\hbar \to 0$ and with fixed eigenvalue $E$, they are  rapidly oscillating in the classically allowed region 
\[\acal_E:=\set{V(x) \leq E},\]
and exponentially decaying in the classically forbidden region 
\[\fcal_E:=\acal_E^c=\set{V(x)>E},\]
with an Airy type transition along the caustic
\[\ccal_E:=\dell \acal_E=\set{V(x)=E}.\]
This reflects the fact that a classical particle of energy $E$  is  confined to $\acal_E=\set{V(x)\leq E}.$ In the forbidden
region, eigenfunctions exhibit exponential decay as $\hbar \to 0$,  measured by the Agmon distance to the caustic. We refer to \cite{Ag,HS} for background. In dimension one, eigenfunctions have no zeros in the forbidden region, but in dimensions $d \geq 2$ they do. In the allowed region, nodal sets of eigenfunctions behave in a similar way to nodal sets on Riemannian manifolds \cite{Jin},  but in the forbidden region they are sparser. The only results at present on forbidden nodal sets seem to be those of \cite{HZZ,CT}. This article contains the first results on the behavior of nodal sets in the transition region around the caustic. The scaling asymptotics of zeros around the caustic is  analogous in many ways to the scaling asymptotics of eigenvalues of $N \times N$  random  Hermitian matrices around the edge of the Wigner distribution in \cite{TW}, and as will be seen, the scaled Airy kernel of \cite{TW} is the same as the scaled
eigenspace projections when $d = 3$ (see Remark \ref{R:TW}).

When $d=1,$ the eigenspaces $V_{\hbar_N, E}$ have dimension $1$ and it is a classical fact (based on WKB or ODE techniques) that Hermite functions and more general \Sch eigenfunctions exhibit Airy asympotics at the caustic (turning points).  See for instance \cite{Sz,O,T,Th,FW}.  The main purpose of this article is to formulate and prove a generalization of these Airy asymptotics to all dimensions for the
isotropic Harmonic Oscillator. Instead of considering individual eigenfunctions,
we consider the scaling asymtptoics of the eigenspace projection
kernels  \eqref{PiDEF} with $x, y$ in an $\hbar^{2/3}$-tube around 
$\ccal_E$. 
Our main result gives scaling asymptotics for the eigenspace projection
kernels \eqref{PiDEF} around a point $x_0 \in \ccal_E$ of the caustic.
To state the result, we introduce some notation. Let $x_0$ be a point on the caustic $|x_0|^2=1$ for $E=1/2$. Points in an $\hbar^{2/3}$ neighborhood
of $x_0$ may be expressed as $x_0 + \hbar^{2/3} u$ with
$u \in \R^d$.  The caustic is a $(d-1)$-sphere whose
normal direction at $x_0$  is $x_0$, so the normal component of $u$ is
 $u_1 x_0$ when $|x_0| = 1$, where $u_1:=\inprod{x_0}{u}$. We also put  $u':=u-u_1x_0$ for the tangential component, and identify $T_{x_0} \ccal_E \cong T^*_{x_0} \ccal_E \cong \R^{d-1}$.   By rotational symmetry, we 
 may assume  $x_0 = (1, 0, \cdots, 0)$, so that $u=(u_1, u_2, \cdots, u_d) =: (u_1; u')$.

\begin{theo}\label{SCLintro}
Let $x_0$ be a point on the caustic $|x_0|^2=1$ for $E=1/2$. Then for  $u,v \in \R^d$,
\begin{equation}
\Pi_{\hbar,1/2} (x_0 + \hbar^{2/3} u, x_0 + \hbar^{2/3} v) = \hbar^{-2d/3+1/3} \Pi_0(u, v) (1 + O(\hbar^{1/3})), \label{E:CausticScaling}
\end{equation}
where  
\be 
\label{Pi0}
\Pi_0(u_1, u'; v_1, v') := 2^{2/3} (2 \pi)^{-d+1} \int_{\R^{d-1}} e^{i \lan u'-v', p\ran } \Ai(2^{1/3}(u_1 + p^2/2))\Ai(2^{1/3}(v_1 + p^2/2)) dp,
\ee
and  $u_1:=\inprod{x_0}{u}$, $u':=u-u_1x_0$ (similarly for $v_1.$)
%\be
% \Pi_0(u, v)  = (2\pi)^{-d/2} 
 % \hbar^{-2d/3+1/3} 
% \int_\ccal T^{-d/2} e^{- \frac{(u-v)^2}{2T} + \frac{T^3}{24} - \frac{T}{2} \lan %u+v, x_0 \ran}  dT \ee
On the diagonal, let $\abs{x}^2 = \abs{x_0+\hbar^{2/3} u}^2 = 1 + \hbar^{2/3} s+O(\hbar^{4/3})$ with $s = 2 \lan x_0, u\ran \in \R$. Then,
\be\label{pi-tube-1} \Pi_{\hbar}(x,x) = 2^{-d+1}\pi^{-d/2} \hbar^{(1-2d)/3} \Ai_{-d/2}(s)(1+O(\hbar^{1/3})). \ee
The error terms in \eqref{E:CausticScaling} and \eqref{pi-tube-1} are uniform when $u,v,s$ vary over a compact set. 
\end{theo}

\noindent 
Above, $\Ai$ is the Airy function, and $\Ai_{-d/2}$ is a {\it weighted Airy function},  defined  for $k \in \R$ by
\begin{equation}
 \label{eq:Ai_k} \Ai_{k}(s) := \int_\ccal T^{k} \exp \left(  \frac{T^3}{3} - T s\right) \frac{dT}{2\pi i},\qquad s\in \R
\end{equation}
where $\ccal$ is the usual contour for Airy function, running from $e^{-i \pi/3} \infty$ to $e^{i \pi/3} \infty$ on the right half of the complex plane (see Appendix A for a brief review of the Airy function).

\begin{remark}\label{R:TW} When $d = 3$, the kernel \eqref{Pi0} with $u' = v'$, i.e.  $\Pi_0(u_1, u'; v_1, u')$, coincides modulo the factor of $\sqrt{\lambda}$ with  the Airy kernel $K(x,y)$ of \cite{TW}
(see (4.5) of that article). The ``allowed region'' of this article is analogous to the `bulk' in random matrix theory, and the ``caustic'' of this article is
analogous to the ``edge of the spectrum''.  \end{remark}

To our knowledge, this is the first result on Airy scaling asymptotics
of \Sch eigenfunctions in dimensions $d > 1$.  Theorem \ref{SCLintro} is proved in Section \ref{TUBESECT}. The on-diagonal result \eqref{pi-tube-1} is proved first in Proposition \ref{pp:caustic} because it is the important case for the applications
to nodal sets.   It is not obvious that  \eqref{Pi0} reduces to \eqref{pi-tube-1} when $u=v$, but this is proved  by combining \eqref{eq:Ai_k} with Lemma \ref{airy-double} on products of Airy functions. The case of general $E$ is obtained by a simple rescaling as in Section \ref{s:notation}.

The isotropic Harmonic Oscillator is special even among Harmonic Oscillators because of the maximally high multiplicity of eigenvalues, and there is no direct
generalization of Theorem \ref{SCLintro} to eigenspace projections of other
\Sch operators.   However, we expect that
 the scaling asymptotics  generalize  if we replace eigenspaces by 
 spectral projections for small intervals in the spectrum (work in progress).  To explain the
 unique features of \eqref{Hh}, we recall  (Section \ref{S:Eigenspaces}) that
 $V_{\hbar_N, E},$ is spanned by Hermite functions of degree $N$ in $d$ variables and
\begin{equation} \label{dimV} \dim V_{\hbar_N,E}= \frac{1}{(d-1)!} N^{d-1}(1 + O(N^{-1})).\end{equation}
The high multiplicities are due to the $U(d)$-invariance of the isotropic Harmonic
Oscillator, and the periodicity of  the  classical Hamiltonian flow. As a
 result, the quantum propagator $e^{-it\widehat{H}_\hbar/\hbar}$ is (essentially) periodic, and
 the  Mehler formula \eqref{E:Mehler} expresses  the propagator  as an integral over the circle. This
 expression is used to obtain the scaling asymptotics of $\Pi_{\hbar, E}(x,y)$ for $x,y$ near $\ccal_E$.   For general Harmonic Oscillators with incommensurate frequencies the eigenvalues have multiplicity one and the eigenspace projections are of a very different type. It is for this reason that we only consider the isotropic Harmonic Oscillator in this article.

{Besides the edge asymptotics of \cite{TW},} the  asymptotic behavior of $\Pi_{\hbar, E}$ in a $\hbar^{2/3}$- neighborhood of the caustic is reminiscent of the scaling asymptotics of the \szego projector in \cite{BSZ} and of spectral projections on Riemannian manifolds \cite{CH}, which both have universal scaling limits. However, the presence of allowed and forbidden regions is a new feature of Schr\"odinger operators that does not occur for Laplacians or in the complex setting. If one rescales $\Pi_{\hbar, E}$  in an $\hbar-$nieghborhood of a point in the allowed region, one would obtain results analogous to those for Laplacians on Riemannian manifolds in \cite{CH}. But the $\hbar^{2/3}$-scaling asymptotics along the caustic are of a fundamentally different nature. Although there are several studies of  Airy asymptotics of Wigner functions in dimension one around the caustic in phase space  (originating in \cite{Be}), we are not aware of any prior studies of the scaling asymptotics of the spectral projections kernels along the caustic in configuration space. It is an important aspect of Schr\"odinger equations that deserves to be studied in generality. 
 In \cite{HZZ3} we study the  Wigner distribution $W_{\hbar, E}$ of \eqref{PiDEF} and its Airy  scaling asymptotics around the phase space energy surface $\Sigma_E = \{\half (|\xi|^2 + |x|^2) = E\}$, which gives
 a higher dimensional generalization of \cite{Be}.

\subsection{Random Hermite eigenfunctions}
Theorem \ref{SCLintro} has several applications to random Hermite
eigenfunctions, which have recently been studied in \cite{HZZ, PRT,IRT}.
A random  Hermite eigenfunction of eigenvalue $E=\hbar(N + \frac{d}{2})$ is defined by
\begin{equation} \label{PHIN} \Phi_{\hbar,E}(x):=\sum_{\abs{\beta}=N}  a_{\beta}\phi_{\beta,\hbar}(x). \end{equation}
where $\beta=(\beta_1, \cdots, \beta_d)$ and $|\beta| = \beta_1 + \cdots + \beta_d$. 
Here the coefficients $a_{\beta}\sim N(0,1)_{\R}$ are i.i.d. normal random variables and $\{\phi_{\beta, \hbar}\}_{\abs{\beta}=N}$ is an orthonormal basis of $V_{\hbar_ N,E}$ consisting of multivariable Hermite functions (see Section \ref{S:Eigenspaces}). Equivalently, we use the basis to identify $V_{\hbar_N, E} \equiv \R^{\dim(V_{\hbar_N, E})}$ and then endow $\R^{\dim(V_{\hbar_N, E})}$ with the standard Gaussian measure. 
%The invariance of the Gaussian under orthogonal transformations shows %that the definition of $\Phi_{\hbar, E}$ is independent of a choice of %orthonormal basis. Note that we have supressed in the notation $\Phi_{\hbar, E}$ its dependence on $N.$
 The Schwartz kernel of the eigenspace projection \eqref{PiDEF}
%\begin{equation} \label{PiDEF} 
%\Pi_{\hbar, E}: L^2(\R^d) \to V_{\hbar_N,E}
%\end{equation}
is the covariance (or two-point) function of $\Phi_{\hbar, E}:$
\begin{equation} \label{COV} 
\Pi_{\hbar,E}(x,y): ={\bf E} ( \Phi_{\hbar, E} (x) \Phi_{\hbar, E}(y)) = \sum_{\abs{\beta}=N} \phi_{\beta,\hbar}(x)\phi_{\beta,\hbar}(y).
\end{equation}  
The random Hermite functions $\Phi_{\hbar,E}$ are a centered Gaussian field. Their properties are therefore completely determined by $\Pi_{\hbar, E}.$

As a first application of Theorem \ref{SCLintro},  we determine  the expected $L^2$ mass of $L^2$-normalized random Hermite eigenfunctions $\Psi_{h, E} := \Phi_{h,E}(x)/ \|\Phi_{h,E}\|_{L^2}$ in a metric tube $T_{\delta}(\ccal_E) =\{y: d(y, \ccal_E) < \delta\}$ of radius $\delta$ around the caustic. Here $d(y, \ccal_E)$ is the distance $\inf_{x \in \ccal_E} |y -x|$ from $y$ to $\ccal_E$. We define the $L^2$-mass-squared of an $L^2$-normalized eigenfunction $\Psi$ in the $\delta$-tube by
\[M_2(\Psi_{\hbar, E}, T_{\delta}(\ccal_E)) := \int_{\tcal_{\delta}(\ccal_E)} |\Psi_{\hbar, E}(x)|^2 dx.\]

\begin{cor}\label{T:Caustic Mass}  Let  $\delta = \kappa \hbar^{2/3}$, $E=1/2$. Then  the expected $L^2$ mass-squared of an $L^2$-normalized Hermite eigenfunction in $ T_{\delta}(\ccal_E)$
  are given by
$$
 \begin{array}{lll}
\E M_2(\Psi_{\hbar, E}, T_{\delta}(\ccal_E))
% &=& C_1   \hbar^{\lr{d-2}/3}\int_{-\kappa/2}^{\kappa/2} \Ai_{-d/2}(u)du  \lr{1+O(\hbar^{1/3})} \\
&=& C(d)\; \hbar^{d/3}\int_{-2\kappa}^{2\kappa}\int_0^\infty \Ai(s+\rho) \rho^{d/2-1} d\rho ds  \lr{1+O(\hbar^{1/3})} 
\end{array}$$
where $C(d) = \frac{\Gamma(d)}{\Gamma(d/2)^2}$.

%{\blue Peng: I put in the correction in blue now. $\kappa/2$ should be $2 %\kappa$, scaling should be $\hbar^{d/3}$ not $\hbar^{(d-2)/3}$, since the %total $L^2$ mass is $1$. The $L^2$ mass density over the $\hbar^{2/3}$ %shell is $\hbar^{(d-2)/3}$, the previous derivation misses a factor of %$\hbar^{2/3}$ in the volume form $(1+\hbar^{2/3} r)^{d-1} d (1+\hbar^{2/3} %r) d \omega \approx \hbar^{2/3} dr d \omega$, where $\omega \in S^{d-1}$. %One should perhaps change the following remark as well.}
\end{cor}

The $L^2$ mass density in the $\hbar^{2/3}$-tube is $\hbar^{(d-2)/3}$;
 integration over the thin tube introduces the additional volume factor  $ \hbar^{2/3}$. The proof is given in \S \ref{CORPROOF}.

%The expected mass is also studied in  \cite{HZZ3}, where the Wigner %distribution of \eqref{COV} is used  instead of propagator kernels.

$L^p$ norms of Hermite eigenfunctions are studied in  \cite{Th,KT} (see also
their references) and pointwise bounds on (non-dilated) Hermite functions are given near the caustic in \cite{Sz}
in dimension 1 (see Lemma 1.5.1 of \cite{Th}).  In dimension 1, the maximum of the $k$th (unscaled) Hermite function is achieved at a point $k^{-1/6}$ close to the caustic (turning points). This motivates the question of how 
$L^p$ mass of eigenfunctions builds up around the caustic for general $p$ in all dimensions. The same question for $p = 1$ arises in the study of nodal volumes. Eigenfunctions concentrated on a single trajectory are extremals for low $L^p$ norms. The above Corollary shows that the mass density is constant when $d = 2$ and decays for $d \geq 3$ for a random Hermite
function.

\begin{remark}
For $d = 1$ it is shown in \cite{Th}, Lemma 1.5.2 that the sup norm
of $L^2$ normalized Hermite functions $H_N$ is $\simeq N^{-\frac{1}{12}}$.
As reviewed in \S \ref{S:Eigenspaces}, the  sup norm of the semi-classically scaled eigenfunctions $\phi_{N, \hbar}$ of this article \eqref{PHI} equal
 $\hbar^{-d/4}$ times the sup norm of the unscaled Hermite functions $H_N$. 
If we fix the energy level $E$ and set
$\hbar (N + \half)= E$ as $\hbar \to 0, N \to \infty$, then
$||\phi_{N, \hbar}||_{L^{\infty}} \simeq \hbar^{-1/4} \hbar^{1/12} E^{-1/12}
\simeq \hbar^{-1/6}$,  which agrees with the mass density formula above.
\end{remark}

\subsection{Applications to nodal sets of random Hermite eigenfunctions}
One of our principal motivations to study the scaling asymptotics of the eigenspace projections is to understand the transition between the behavior of
nodal sets of random eigenfunctions $\Phi_{\hbar, E}$ in the allowed and forbidden regions. In particular, we study in 
\begin{figure}{rt}

\vspace{-11pt}
\begin{center}
  \includegraphics[width=.4 \textwidth]{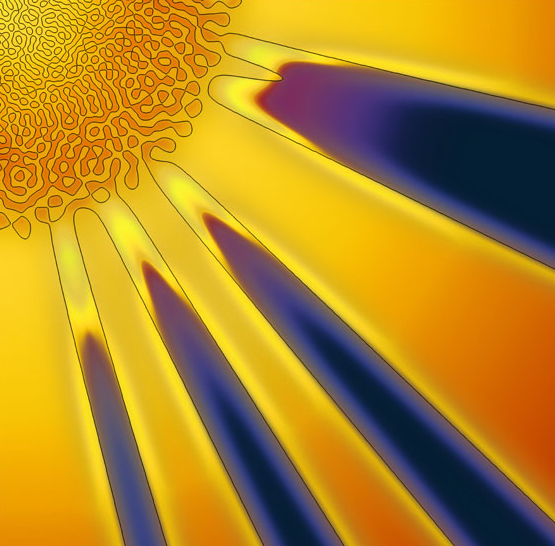} 
\end{center}
\caption{Nodal line (colored black) of $\Phi_{\hbar, E},$ $d=2.$ (Taken from \cite{BH})\label{F:Heller}}
\vspace{-10pt}
\end{figure}
Theorem \ref{CAUSTIC} the average density of the nodal set 
\[Z_{\Phi_{\hbar, E}}:=\set{\Phi_{\hbar, E}(x)=0}\] 
near $\ccal_E.$ Let us denote by $\ZN$ the random measure of integration over $Z_{\Phi_{\hbar, E}}$ with respect to the Euclidean hypersurface measure (the co-dimension $1$ Hausdorff measure $\hcal^{d-1}$) of the nodal set. Thus for any measurable set $B \subset \R^d$,
$$\ZN (B) := \hcal^{d-1} (B \cap Z_{\Phi_{\hbar, E}}).$$
Its expectation $\E \ZN$ is the average density or distribution of zeros and is the measure on $\R^d$ defined by
$$\E \ZN (B) = \int_{V_{\hbar, E}}  \hcal^{d-1} (B \cap Z_{\Phi_{\hbar, E}}) d\gamma_{\hbar, E},$$
where $\gamma_{\hbar, E}$ is the Gaussian from which $\Phi_{\hbar, E}$ is sampled. In recent work \cite{HZZ}, the authors showed that $\E \ZN$ has a different order as $\hbar \to 0$ in the allowed and forbidden regions. 

\begin{theo}[\cite{HZZ} 2013]  \label{THZZ} 
Fix $E>0.$ The measure $\E\ZN$ has a density $F_{\hbar, E}$ with respect to Lebesgue measure given by
\begin{align}
&F_{\hbar, E}(x) = \hbar^{-1}\cdot c_d \sqrt{2E-\abs{x}^2}\lr{1+O(\hbar)}\label{E:Allowed Density1}, &  \mbox{ if}~x\in \acal_E\backslash \{0\}\\
&F_{\hbar, E}(x) = \hbar^{-1/2}\cdot C_d \frac{E^{1/2}}{\abs{x}^{1/2}\lr{\abs{x}^2-2E}^{1/4}}\lr{1+O(\hbar)},&\mbox{ if}~ x\in \fcal_E, 
\label{E:Forbidden Density1}
\end{align}
where the implied constants in the `$O$' symbols are uniform on compact subsets of the interiors of $\acal_E\backslash \{0\}$ and $\fcal_E$, and where $c_d=\frac{\Gamma\lr{\frac{d+1}{2}}}{\sqrt{d\pi}\Gamma\lr{\frac{d}{2}}}$ and $C_d=\frac{\Gamma\lr{\frac{d+1}{2}}}{\sqrt{\pi}\Gamma\lr{\frac{d}{2}}}$ depend only on $d.$ 
%\[c_d=
%\qquad
%\text{and}\qquad C_d =
%\frac{\Gamma\lr{\frac{d+1}{2}}}{\sqrt{\pi}\Gamma\lr{\frac{d}{2}}}.\]
\end{theo}

Our main result on nodal sets (Theorem \ref{CAUSTIC}) gives scaling asymptotics for the average nodal density that `interpolate' between \eqref{E:Allowed Density1} and \eqref{E:Forbidden Density1}. The computer graphics of Bies-Heller \cite{BH} (reprinted as Figure \ref{F:Heller} in \cite{HZZ}) show that the nodal set in $\acal_E$ near the caustic $\partial \acal_E$ consists of a large number of highly curved nodal components apparently touching the caustic while the nodal set in $\fcal_E$ near $\partial \acal_E$ consists of fewer and less curved nodal components all of which touch the caustic. The scaling limit of the density of zeros in a shrinking neighborhood of the caustic, or in annular subdomains of $\acal_E$ and $\fcal_E$ at shrinking distances from the caustic are given Theorems \ref{CAUSTIC} and \ref{T:Mainb}. The varying density of zeros in $\acal_E, \fcal_E$ near the caustic proved there is the new phenomenon related to nodal sets at issue in this article.    For the expected density of zeros in the $\hbar^{2/3}$-neighborhood of the caustic, we apply the Kac-Rice formula to Theorem \ref{SCLintro};  for the expected density of zeros in the $\hbar^{\alpha}$-neighborhood of the caustic, where $0<\alpha <2/3$, we invoke the Kac-Rice formula together with a non-standard stationary phase method (see Propositions \ref{pp:allowed} and \ref{pp:forbidden}).

%\subsection{Re-scaled nodal set in shrinking balls around caustic points} 
The nodal set of $\Phi_{\hbar, E}$ near the caustic consists of a mixture of components from the forbidden nodal set and from the allowed nodal set (see Figure \ref{F:Heller}). To be more precise, if $\psi\in V_{\hbar, E}$ is non-zero, then there do not exist nodal domains contained entirely in $\fcal_E$, where the potential $V:=\abs{x}^2/2$ is greater than the energy $E$, because $\Delta \psi = (V - E) \psi$ forces $\psi$ and $\Delta \psi$ to have the same sign in $\fcal_E$. In a nodal domain $\dcal$ we may assume $\psi > 0$, but then $\psi$ is a positive subharmonic function in $\dcal$ and cannot be zero on $\partial \dcal$ without vanishing identically. Hence, every nodal component which intersects $\fcal_E$ must also intersect $\acal_E$ and therefore $\ccal_E$. As indicated in the computer graphics of \cite{BH}, some of these components remain in a very small tube around $\ccal_E$ and some stretch far out into $\fcal_E$ and most of these (apparently) stretch out to infinity. The density results of Theorem \ref{THZZ} suggest that there should only exist on average order of $\hbar^{-1/2}$ of the latter, not enough to explain the $\hbar^{-\frac{2}{3}}$ size of $\E \ZN$ around the caustic proved in Theorem \ref{CAUSTIC} below (see also Remark \ref{R:Unscaled}). Hence, one expects the main contribution to the nodal density near $\ccal_E$ to come from nodal components living mainly in $\acal_E$ which cross $\ccal_E$. 

To state our first result precisely, we fix $x \in \ccal_E$, where $E=1/2$, and study the  rescaled ensemble
\[\Phi_{\hbar, E}^{x,\alpha}(u):=\Phi_{\hbar,E}(x+\hbar^\alpha u)\]
and the associated hypersurface measure 
\[\abs{Z_{\hbar, E}^{x,\alpha}}(B)=\hcal^{d-1}\lr{\set{\Phi_{\hbar, E}^{x,\alpha}(v)=0}\cap B},\qquad B\subset \R^d.\]
Our main result gives the asymptotics of $\E \abs{Z_{\hbar, E}^{x,\alpha}}$ when $\alpha = 2/3$ is in terms of the weighted Airy functions $\Ai_k$ (see \eqref{eq:Ai_k}).

\begin{theo}[Nodal set in a shrinking ball around a caustic point]\label{CAUSTIC} Fix $E=1/2$ and $x\in \mathcal C_E$, i.e. $|x|=1$. For any bounded measurable $B\subseteq \R^d,$ 
\[\E \abs{Z_{\hbar, E}^{x,2/3}}(B)=\int_B \fcal(u)du,\]
where 
\be \fcal(u)= \lr{2\pi}^{-\frac{d+1}{2}}\int_{\R^d}|\Omega(u)^{1/2}\xi|e^{-\abs{\xi}^2/2}d\xi ~(1 + O(\hbar^{1/3}))\label{E:CausticDensity}\ee
and $\Omega=\lr{\Omega_{ij}}_{1\leq i,j\leq n}$ is the symmetric matrix
\begin{equation}
\Omega_{ij}(u) = x_i x_j \left( \frac{\Ai_{2-d/2}(s)}{\Ai_{-d/2}(s)} - \frac{\Ai^2_{1-d/2}(s)}{\Ai^2_{-d/2}(s)} \right) + \frac{\delta_{ij}}{2} \frac{\Ai_{-1-d/2}(s)}{\Ai_{-d/2}(s)}.\label{OMEGAintro}
\end{equation}
where $s = 2\lan u, x \ran$. 
The implied constant in the error estimate from \eqref{E:CausticDensity} is uniform when $u$ varies in compact subsets of $\R^d$. 
\end{theo}
\begin{remark}
The leading term in $\fcal$ is $\hbar$-independent and positive everywhere since the matrix $\Omega_{ij}(u)$ as a linear operator has nontrivial range. Indeed, as shown in Proposition \ref{Ai-pos}, $\Ai_{-d/2}(s) > 0$ for all integers $d \geq 2$, hence $\delta_{ij}$ term in \eqref{OMEGAintro} is nonzero. The matrix $\lr{x_i x_j}_{i,j}$ in \eqref{OMEGAintro} is a rank $1$ projection onto the $x-$direction. Hence, since the dimension $d \geq 2$, it cannot cancel out the second term.
\end{remark}
\begin{remark}\label{R:Unscaled}
Theorem \ref{CAUSTIC} says that if $x\in \ccal_E$ and $\twiddle{B}_\hbar=x+\hbar^{2/3}B$ for some bounded measurable $B,$ then 
\[\E{\abs{Z_{\Phi_{\hbar,E}}}}(\twiddle{B}_\hbar)=\hbar^{2/3\lr{d-1}}\E{\abs{Z_{\hbar, E}^{x,\alpha}}}(B)=\hbar^{-2/3}\int_{\twiddle{B}_\hbar}\fcal(\hbar^{-2/3}\lr{y-x})dy,\]
which shows that the average (unscaled) density of zeros in a $\hbar^{2/3}-$tube around $\ccal_E$ grows like $\hbar^{-2/3}$ as $\hbar\gives 0.$
\end{remark}
The choice of radius $\delta = \hbar^{\frac{2}{3}}$ is dictated by the scaling asymptotics of the associated covariance (2-point) function \eqref{COV}, which are stated in Theorem \ref{SCLintro} and proved in  \S \ref{TUBESECT}. The rescaled random Hermite functions $\Phi_{\hbar, E}^{x,2/3}(v)$ converge to an infinite dimensional Gaussian ensemble of solutions of a scaled eigenvalued problem, which is identified in Section \ref{LTENSEMBLE}.

\begin{remark} 
%The scaled Airy kernel \eqref{Pi0} is used to determine
%the density of scaled zeros around the caustic in Theorem \ref{CAUSTIC}. 
As mentioned above, the scaling asymptotics of zeros around the caustic,
especially in the radial (normal) direction,   is analogous to the scaling asyptotics of  eigenvalues of
random Hermitian matrices around the edge of the spectrum. But the scaled
radial distribution of zeros 
of random Hermite eigenfunctions does not seem to be a determinantal process, while
 eigenvalues of random Hermitian matrices is determinantal.
\end{remark}

\subsection{Sub-critical shrinking of balls around caustics points}\label{ANNULI}  So far, we have considered in detail the rescaling of the eigenspace projections and nodal sets in a $\hbar^{2/3}$-tube around the caustic. In \cite{HZZ} we have studied the bevavior within the allowed and forbidden regions at fixed ($\hbar$-independent) distance from the caustic. We refer to such regions as the `bulk'.  In  section \ref{ANNULISECT}, we fill in the gaps between the caustic tube and the `bulk' in the allowed and forbidden regions. That is, we consider shrinking annuli around the caustic which lie outside the $\hbar^{2/3}$-tube in a sequence of rings around the caustic. In this way, we obtain scaling results that interpolate between the bulk results of \cite{HZZ} and the caustic scaling results above.
This is the purpose of studying sub-critical rescaling exponents, i.e. the nodal set around $x + \hbar^\alpha v$ where $x \in \ccal_E$ and $0<\alpha<2/3$.%In short, 
 
\begin{theo}\label{T:Mainb}
Fix $E=1/2$ and $x \in \ccal_E$, and $\alpha<2/3$. Then the rescaled expected distribution of zeros $\E |Z^{x, \alpha}_{\hbar, E}|(u)$ has a density $F^{x, \alpha}_{\hbar, E}(u)$ with respect to Lebesgue measure given by
\begin{align}
&F^{x, \alpha}_{\hbar, E}(u) = \hbar^{-(1-\frac{3}{2}\alpha)} c_d \cdot s^{1/2} \lr{1 + O(\hbar^{\alpha/2}) + O(\hbar^{1-(3/2)\alpha})} \label{E:Allowed Density 2},   \mbox{ if}~x+\hbar^\alpha u\in \acal_E\backslash \{0\}\\
&F^{x, \alpha}_{\hbar, E}(u) = \hbar^{-\half(1-\frac{3}{2}\alpha)}    C_d \cdot s^{-1/4} \lr{1 + O(\hbar^{\alpha/2}) + O(\hbar^{1-(3/2)\alpha})},\mbox{ if}~ x+\hbar^\alpha u\in \fcal_E ,
\label{E:Forbidden Density 2}
\end{align}
where $s=2\inprod{u}{x},$ the implied constants in the `$O$' symbols are uniform for $u$ in a compact subset of $\R^d$,  and $c_d,C_d$ are positive dimensional constants. 
\end{theo}
\begin{remark}
As in Remark \ref{R:Unscaled}, the rescaled density $F_{\hbar, E}^{x,\alpha}$ has an extra $\hbar^\alpha$ factor compared with the unscaled density $\fcal_{\hbar, E}(x+\hbar^\alpha u)$ from Theorem \ref{THZZ}. Hence, writing $d(y,\ccal_E)$ for the distance from $y$ to the caustic, equations \eqref{E:Allowed Density 2} and \eqref{E:Forbidden Density 2} agree with the results of Theorem \ref{THZZ}
 \begin{equation*}
 \begin{cases}
 F_{\hbar,E}(x+\hbar^\alpha u) \sim \hbar^{-1} d\lr{x+\hbar^\alpha u,\mathcal C_E}^{1/2}  \sim \hbar^{-1+\alpha/2},& x\in \acal_E\\
 F_{\hbar,E}(x+\hbar^\alpha u) \sim \hbar^{-1/2} d\lr{x+\hbar^\alpha u,\mathcal C_E}^{-1/4} \sim \hbar^{-1/2-\alpha/4},& x\in \fcal_E
\end{cases}
\end{equation*}
since $-1+\alpha/2+\alpha = -(1-3\alpha/2), -1/2-\alpha/4+\alpha = -\half(1-3\alpha/2)$.
\end{remark}

\begin{remark}
The two error terms $O(\hbar^{\alpha/2}), O(\hbar^{1-(3/2)\alpha})$ come from the non-standard stationary phase expansion used to prove Theorem \ref{T:Mainb}: the $O(\hbar^{\alpha/2})$ comes from approximating the critical point of the phase function, while the $O(\hbar^{1-(3/2) \alpha})$ marks the failure of the stationary phase expansion as $\alpha$ approaches $2/3$. 
\end{remark}
 The main ingredient in the proof is the sub-critical scaling asymptotics of the covariance function \eqref{COV} in Propositions \ref{pp:allowed} resp. \ref{pp:forbidden} around points in $\acal_E$, resp. $\fcal_E$.

\subsection{Nodal set intersections with the caustic} Our next result measures the density of intersections of the nodal set with the caustic. This is much simpler than measuring the density in shrinking tubes or annuli, since it is not necessary to rescale the covariance kernel  \eqref{COV}. For an open set  $B \subset \ccal_E$ we consider 
$$ \hcal^{d-2} (B \cap Z_{\Phi_{\hbar, E}}).$$
When $d = 2$ this means to count the number of nodal intersections with the caustic. Since the Gaussian measure is $SO(d)$ invariant and the caustic is a sphere, the average nodal density along the caustic is constant. 
%In dimension 2, the expected number of nodal intersections is especially concrete:

\begin{theo} \label{T:Caustic Zeros} Fix $d \geq 2,$ $E=1/2,$ and define the constant
\[\fcal_{\ccal_E,d} = \frac{\Gamma\lr{\frac{d}{2}}}{\sqrt{2\pi} \Gamma\lr{\frac{d-1}{2}}}\cdot \lr{\frac{\Ai_{-1-d/2}(0)}{\Ai_{-d/2}(0)}}^{1/2},\]
where $\Ai_k$ are defined in \eqref{eq:Ai_k}. Then, as $\hbar \gives 0,$ for any open set  $B \subset \ccal_E$
\[\E  \hcal^{d-2} (B \cap Z_{\Phi_{\hbar, E}}) = \hbar^{-\frac{2}{3}}  \fcal_{\ccal_E,d}\vol_{S^{d-1}}(B) \lr{1+O(\hbar^{1/3})}.\] 
In particular, if $d = 2$,  $\E \left( \# Z_{\Phi_{\hbar, E}} \cap \ccal_E \right) = C_0 \; \hbar^{-2/3}\lr{1+O(\hbar^{1/3})} $ where
\[ C_0 = \sqrt{2} \lr{\Ai_{-2}(0) / (\Ai_{-1}(0)) }^{1/2}>0.\]
\end{theo}

The proof is to use a  Kac-Rice formula for the expected number of intersections of the nodal set with the caustic. The relevant covariance kernel is the restriction  of $\Omega_{ij}$ in \eqref{OMEGAintro} to the tangent plane of the caustic (hence the radial component $x_ix_j$ drops out). It is analogous to the formula in \cite{TW} (Proposition 3.2) for the expected number of intersections of nodal lines with the boundary of a plane domain, and we therefore omit  its proof.

%\subsection{\label{SLINTRO}Scaling limit around the caustic} 

%The principal ingredient in Theorem \ref{CAUSTIC} is the scaling asymptotics of \eqref{COV} around the caustic. It is the main analytical result of this article. 
 
%\begin{theo}\label{SCLintro}
%Let $x_0$ be a point on the caustic $|x|^2=1$ for $E=1/2$, then
%\be
%\Pi_\hbar (x_0 + \hbar^{2/3} u, x_0 + \hbar^{2/3} v) = \hbar^{-2d/3+1/3} \Pi_0(u, v) (1 + O(\hbar^{1/3})) 
%\ee
%where  
%\be 
%
%\be
% \Pi_0(u, v)  = (2\pi)^{-d/2} 
 % \hbar^{-2d/3+1/3} 
% \int_\ccal T^{-d/2} e^{- \frac{(u-v)^2}{2T} + \frac{T^3}{24} - \frac{T}{2} \lan %u+v, x_0 \ran}  dT \ee
%\end{theo}
%See \eqref{Pi0} for a formula in terms of Airy functions. 
%The proof is given in Section \ref{TUBESECT}. The result on the diagonal is proved first in Proposition \ref{pp:caustic}.

\subsection{Outline of the proofs}\label{S:Outline of Proofs}
As mentioned above, the proofs of Theorems \ref{CAUSTIC}, \ref{T:Mainb}, and \ref{T:Caustic Zeros} are based on a detailed analysis of the Kac-Rice formula (Lemma \ref{L:Gaussian KR} in \S \ref{KRSECT}), which gives a formula for the average density of zeros at $x$ in terms of a covariance matrix $\Omega_x$ depending only on $\Pi_{\hbar, E}(x,y)$ and its derivatives
\[d_x  \Pi_{\hbar, E}(x,y), d_y \Pi_{\hbar, E}(x,y), d_x d_y \Pi_{\hbar, E}(x,y)\]
restricted to the diagonal. The kernels behave differently depending on the position of $x$ relative to the caustic and have contour integral representations of the following type
\[ \int_\ccal \frac{1}{z^k} \exp \left({\frac{ -u z }{\hbar^{1-\alpha}}+ \frac{z^3}{\hbar}}\right) dz \]
where $\ccal$ is the usual Airy contour (Appendix A), $u\in\C$ is a parameter, and we simplified the phase function and the amplitude by keeping only the leading term. See \eqref{eq:Pi_Contour} for the precise formulas. The critical points of the simplified phase function are $z_{\pm} = \pm \sqrt{u/3} \hbar^{\alpha/2}$.  If $\alpha=0$, we would have two separate critical points, which were analyzed in our previous paper \cite{HZZ}. If $0< \alpha <2/3$, then as $\hbar \to 0$ the two critical points $z_\pm$ starts to move closer, with  $|z_+-z_-| \sim \hbar^{\alpha/2}$, and each corresponds with a Gaussian bump of width $\hbar^{1/2-\alpha/4}$. Hence the two bumps will start overlapping when $\alpha = 2/3$. In more details, the integral near each critical point can be roughly evaluated as follows 
\bea &&C_\hbar^{-1} \int (z_c + \eta)^{-k} e^{ - \frac{\eta^2}{\hbar^{1-\alpha/2}}} d\eta =   C_\hbar^{-1} \int ((z_c)^{-k} +C(k) (z_c)^{-k-2} \eta^2 + \cdots)   e^{ - \frac{\eta^2}{\hbar^{1-\alpha/2}}} d\eta \\
& =&(z_c)^{-k} (1 +C(k) (z_c)^{-2} \hbar^{1-\alpha/2} + \cdots) = (z_c)^{-k} (1 + O(\hbar^{1-(3/2)\alpha}))
\eea
where $C_\hbar =\int e^{ - \frac{\eta^2}{\hbar^{1-\alpha/2}}} d\eta$ is a normalization constant. The point of the above sketch computation is to show that due to the singularity of $z^{-k}$, the error term is enhanced to $O(\hbar^{1-(3/2)\alpha})$ from the naive expectation $O(\hbar/Hess(\Phi)) = O(\hbar^{1-\alpha/2})$. 
Hence the break down of the expansion at $\alpha=2/3$. Another way to see the break down of stationary phase at $\alpha=2/3$, which we do not go into detail in the paper\footnote{The reason we did not adopt the rescaling $z = \hbar^{\alpha/2} \eta$ is that, in the actual calculation, we have higher order terms $z^4, z^5$ etc in the phase function, and after the change of variable we get $e^{\hbar^{-1} z^4} = e^{\hbar^{-1+2\alpha} \eta^4}$, which cannot be treated perturbatively unless $\alpha > 1/2$. To use this method, one needs to use the Malgrange preparation theorem (see \cite{FW}) to get rid of the higher order terms, which we choose not to do in this paper.} , is to rescale $z$ centered at $z=0$ rather than at $z=z_c$. Namely if we set $z = \hbar^{\alpha/2} \eta$, then we get
\[ \int \eta^{-k} \exp \left({\frac{ -u \eta + \eta^3}{\hbar^{1-(3/2) \alpha}}}\right) d\eta \]
Hence we see at $\alpha = 2/3$, the $\hbar$ factor in the exponent is unity, hence we do not have a small parameter to do asymptotic expansion, and have to express the result in terms of the weighted Airy functions $\Ai_{-k}$ defined in \eqref{eq:Ai_k}.

%This relation with the Airy function may be seen in two ways, one from the expression of $\Pi_{\hbar, E}$ in terms of the propagator \eqref{eq:pixy}, and the other from the relation of $\Pi_{\hbar, E}(x,x)$ to its so-called Wigner distribution $W_{\hbar, E}(x, \xi)$, which may be expressed in terms of the Laguerre function. In the caustic region, the Laguerre function and hence the Wigner function $W_{\hbar, E}(x, \xi)$ have asymptotics in terms of Airy functions. Since the analysis of the Wigner function is lengthy and of independent interest, we carry it out in the related article \cite{HZZ3}, along with a discussion of Wigner functions of spectral projections for general Harmonic Oscillators.

\subsection*{Notation\label{s:notation}}
To get rid of  $2E$ factors, we will set $2E=1$ for the rest of the article, and drop the $E$ subscript. The general case can be obtained by the following replacement
\begin{equation}
 \label{eq:scaling} x \mapsto x':= x/\sqrt{2E} ,\quad p \mapsto p' :=p/\sqrt{2E}, \quad \hbar \mapsto \hbar':= \hbar/2E
\end{equation}
and 
\[ \Pi_{\hbar, E}(x,y) = (2E)^{-d/2}  \Pi_{\hbar'}(x',y'), ~~~\Omega_{\hbar, E}  = (2E)^{-1} \Omega_{\hbar'}(x'), ~~~\; F_{\hbar,E}(x) =  (2E)^{-1/2} F_{\hbar'}(x'). \]
For general $E$, the dilation factor $\delta=C\hbar^\alpha$ in Theorem \ref{T:Mainb} should be changed to $C (\hbar/2E)^\alpha$. We will also abbreviate throughout $V_{\hbar_N, E}=V_{\hbar, E}$ with the understanding that $E$ is fixed so that $\hbar=\hbar_N=E/N-d/2.$

\section{\label{BACKGROUND} Background}
We follow the notation of \cite{HZZ} and refer there for background on
the isotropic Harmonic Oscillator, its spectrum and its spectral projections. 
We also refer to that article for background on the Kac-Rice formula and
other fundamental notions on Gaussian random Hermite eigenfunctions.
In this section, we recall some of the basic definitions and facts that are
used in the proofs of the main results.

%Two spectral functions play a fundamental role in the proofs:

%\begin{itemize}

%\item The most important spectral function is the covariance kernel or %spectral projections kernel $\Pi_{\hbar, E}$ \eqref{COV}, whose properties %are recalled in \S \ref{PISECT};
%\bigskip

%\item The propagator $U_h(t, x, y)$, which is used to give a convenient %expression \eqref{E:Projector Integral Forbidden} for $\Pi_{\hbar, E}$ (see %\S \ref{USECT}); \bigskip

%\end{itemize}

%\subsection{An orthonormal basis}

\subsection{Eigenspaces}\label{S:Eigenspaces}

Fix $E = 1/2$. An orthonormal basis of  the eigenspace $V_{\hbar_N, E}$ \eqref{VNE}
%(the $E=\hbar\lr{N+d/2}$-eigenspace of $\widehat{H}_{\hbar}$) 
is given by 
%{\red Peng: It seems wrong to have $2E$ factors in $p_{\beta}\lr{2 E \hbar^{-1/2}\; x}$. The formula is for $2E=1$. To recover the general case, one can use the replacement rule: then $x/\sqrt{\hbar}$ changes to $(x/\sqrt{2E}) / \sqrt{ \hbar/2E} = x / \sqrt{\hbar}$, it turns out the $2E$ factors cancels out. So I removed it, and added $E=1/2$ in the begining.} 
\begin{equation}
\phi_{\beta,\hbar}(x)=\hbar^{-d/4}p_{\beta}\lr{\hbar^{-1/2}\; x
}e^{-x^2/2h},\label{PHI}
\end{equation}
where $\beta=\lr{\beta_1,\ldots, \beta_d}\geq (0,\ldots,0)$ is a $d-$dimensional multi-index with $|\beta|=\sum_i \beta_i = N$ and $p_{\beta}(x) = \prod_{j=1}^d p_{\beta_j}(x_j)$ is the product of the Hermite polynomials $p_k$ (of degree $k$) in one variable, with the normalization that $\int_\R p_k(x)^2 e^{-x^2} dx = 1$. The eigenvalue of $\phi_{\beta,\hbar}$ is given by
\begin{equation} \label{EV}
H_\hbar \phi_{\beta,\hbar} = \hbar (|\beta|+d/2) \phi_{\beta,h}.
\end{equation}
The multiplicity of the eigenvalue $ \hbar (|\beta|+d/2)$ is the partition function of $|\beta|$, i.e. the number of $\beta=\lr{\beta_1,\ldots, \beta_d}\geq (0,\ldots,0)$  with a fixed value of $|\beta|$. Hence 
\[ \dim V(\hbar, E) = {N+d-1 \choose d-1 } = \frac{1}{(d-1)!} N^{d-1}(1 + O(N^{-1}) .\]
For further background and notation we refer to \cite{HZZ}.

\subsection{\label{USECT} Mehler Formula for the propagator}
The Mehler formula gives an explicit expression for the Schwartz (Mehler) kernel 
\[ U_\hbar (t,x,y) = e^{-\frac{i}{\hbar} t \widehat{H}_\hbar}(x,y) = \sum_{N=0} ^\infty e^{-\frac{i}{\hbar} t E_N} \Pi_{\hbar, E_N}(x,y)\] 
of the propagator, $e^{-\ihbar t \widehat{H}_\hbar}.$ The Mehler formula \cite{F} reads
\begin{equation}
 U_\hbar (t, x,y) =e^{-\ihbar t \widehat{H}_\hbar}(x,y)= \frac{1}{(2\pi i \hbar \sin t)^{d/2}}
 \exp\left( \frac{i}{\hbar}\left(
 \frac{\abs{x}^2 + \abs{y}^2}{2} \frac{\cos t}{\sin t} - \frac{x\cdot
 y}{\sin t} \right) \right),
 \label{E:Mehler}
\end{equation}
where $t \in \R$ and $x,y \in \R^d$. The right hand side is singular at $t=0.$ It is well-defined as a distribution, however, with $t$ understood as $t-i0$. Indeed, since $\widehat{H}_\hbar$ has a positive spectrum the propagator $U_\hbar$ is holomorphic in the lower half-plane and $U_\hbar (t, x, y)$ is the boundary value of a holomorphic function in $\{\Im t < 0\}$.

\subsection{\label{PISECT} Spectral projections}
We also use that the spectrum of $\widehat{H}_\hbar$ is easily related to the integers $|\beta|$. The operator with the same eigenfunctions as $\widehat{H}_\hbar$ and eigenvalues $\hbar |\beta|$ is often called the number operator, $\hbar\ncal$. If we replace $U_\hbar (t)$ by $e^{- \frac{i t}{\hbar} \ncal}$ then the spectral projections $\Pi_{\hbar, E}$ are simply the Fourier coefficients of $e^{- \frac{i t}{\hbar} \ncal}$. In \cite{HZZ}, we used the related formula,
\begin{align}
\label{E:Projector Integral Forbidden}
\Pi_{\hbar, E}(x,y)&=\int_{-\pi-i\epsilon}^{\pi-i\epsilon} U_\hbar (t,x,y) e^{\ihbar t E} \frac{dt}{2\pi},
\end{align}
where, as before, $E=\hbar\lr{N+d/2}.$ The integral is independent of $\epsilon$. Using the Mehler formula \eqref{E:Mehler} we obtain a rather explicit integral representation of \eqref{COV}. If we introduce a new complex variable $z = e^{-it}$, then the above integral can be written as
\begin{equation}
\label{eq:pixy}
\Pi_{\hbar, E_N}(x,y) = \oint_{C_\epsilon}\left(\frac{z}{\pi \hbar (1-z^2)} \right)^{d/2}e^{-\frac{1}{\hbar} \left( \frac{1+z^2}{1-z^2}\frac{x^2+y^2}{2} - \frac{2z}{1-z^2} x \cdot y + E\log z \right)}  \frac{dz}{2\pi i z} 
\end{equation}
where the contour is a circle $C_\epsilon=\set{\abs{z}=1- \epsilon}$ traversed counter-clockwise.

%\subsection{Kac-Rice formula  near the caustic}\label{TMainbSECT}

\subsection{\label{KRSECT} The Kac-Rice Formula} As with Theorem \ref{THZZ}, the proofs of Theorems \ref{CAUSTIC}, \ref{T:Mainb} and \ref{T:Caustic Zeros} are based on the Kac-Rice formula \cite[Thm. 6.2, Prop. 6.5]{AW} for the average density of zeros.  The Kac-Rice formula is  the formula for the pushforward of the Gaussian measure on the random Hermite functions under the evalution maps $ev_x(\phi) = (\phi(x), \nabla \phi(x))$. It is valid at $x\in \R^d$ as long as the so-called 1-jet spanning property holds at $x.$ Namely, $ev_x: V_{h_N, E} \to \R^{d +1}$ is surjective, or equivalently, the $(d +1) \times (d +1)$ covariance matrix  $\Sigma_x = \Cov(\Phi_{\hbar, E}(x), \, \nabla_x\Phi_{\hbar,E}(x))$ of values and gradients is invertible. We now verify that the condition for the validity of the Kac-Rice formula  holds.

\begin{prop}
  For any $x \in \R^d \RM \{0\}$, the 1-jet evaluation map
  \[ ev_x: V_{h_N, E} \to \R^{d +1}, \quad \varphi \mapsto  (\varphi(x), \pa_1 \varphi(x), \cdots, \pa_d \varphi(x)) \]
  is surjective, where $V_{h_N, E}$ is defined in Eq \eqref{VNE}. Equivalently, if $V_{h_N, E}$ is equipped with a
  standard Gaussian measure $\gamma_{V}$ induced by the inner product on $V_{h_N, E}$, then its pushforward under ${ev_x}_*$ is a non-degenerate Gaussian measure on $\R^{d+1}$.
\end{prop}
\begin{proof}
Fix an orthonormal basis $\{\phi_j\}_{j=1}^{d_N}$ of $V_{h_N, E}$, then $ev_x$ can be written as a $(d+1) \times d_N$ matrix $M_x$, where $d_N = \dim V_{h_N, E}$. Showing $ev_x$ is surjective is equivalent to showing $M_x$ has rank $d+1$, or $\Sigma_x := M_x M_x^t$ is a non-degenerate $(d+1) \times (d+1)$ square matrix. By definition, $\Sigma_x$ is the covariance matrix of the Gaussian measure ${ev_x}_*(\gamma_V)$. Hence, the two statements in the proposition are equivalent. 

Recall that $\Phi_{h,E}$ is the Gaussian random variable valued in $V_{h_N, E}$ with measure $\gamma_V$.  We then express $\Sigma_x$ in polar coordinates  $x=(r,\omega)$ where $\omega \in S^{d-1}$ and $r \in \R_{>0}$. The first observation is that  $\Sigma_x$ is block diagonal if we break it up into its radial part and angular part,
\begin{equation}\Sigma_{(r,\w)}  =
\left.\left(\begin{array}{cc} 
 \Cov\lr{\Phi_{\hbar, E}, \dell_r \Phi_{\hbar, E}} & 0\\
0 & \Cov(\nabla_\w \Phi_{\hbar,E},\nabla_\w \Phi_{\hbar,E})
\end{array}\right) \right \vert_{(r,\w)}.\label{E:1jet}
\end{equation} 
Indeed, the $(\phi, \nabla_{\omega} \phi)$ block is $\nabla_{\omega} \Pi_{h,E}(r, \omega, r, \omega') |_{\omega =\omega'}= 0$ and $\Pi_{h, E} (x,x)$ is invariant under rotations in $SO(d)$. For the same reason the mixed $(r, \omega)$ deriviatives are zero. The block-diagonality may be expressed in an invariant form by combining the second derivative block  into the  Riemannian metric 
\begin{equation} \label{gx} g_x: = d_x \otimes d_y \Pi_{h_N, E} (x, y) |_{x = y} = \sum_{\abs{\beta}=N} d \phi_{\beta,\hbar}(x) \otimes d\phi_{\beta,\hbar}(x) \end{equation} of the process.  Due to the $O(d)$ symmetry, the metric has the form $G(r)\; dr^2 + H(r) \; g_{S^{d-1}}$ where $g_{S^{d-1}}$ is the standard metric of $S^{d-1}$
and $G(r), H(r)$ are radial functions. This is equivalent to the statement that the  angular block is
orthogonal to the radial block.

Next we check that the angular derivative block is invertible, i.e. that $H(r) > 0$.  The isotropy group of $x$ is  $SO(d-1)$ acting in the tangent space to the sphere centered at the origin through $x$. By the $SO(d-1)$ symmetry
\[  \left. \Cov(\nabla_\w \Phi_{\hbar,E},\nabla_\w \Phi_{\hbar,E}) \right|_{(r,\w)} = H(r) \cdot I_{d-1} \]
where $I_{d-1}$ is the diagonal matrix.  Taking trace on both sides of the above equation, we get
\[ (d-1) H(r) = \sum_{j=1}^{d_N} |\nabla_\omega \phi_j(r,\w)|^2.\]
If $H(r)=0$ for some $r=r_0$, then it means every Hermite function in $V_{h_N, E}$ is constant at the sphere with radius $r=r_0$, which is absurd since any product Hermite functions in Eq. \eqref{PHI} is not constant at any sphere with positive radius. To complete the proof, we need to show that the upper  $2 \times 2$ 
 block
\[\Cov\lr{\Phi_{\hbar, E}, \dell_r \Phi_{\hbar, E}}  
 = \begin{pmatrix} \Pi_{\hbar, E}(x,x) & \half \partial_r \Pi_{\hbar, E}(x,x) \\ &\\
  \half \partial_r \Pi_{\hbar, E}(x,x)  & \partial_r \otimes \partial_{r'} \Pi_{\hbar, E}(x,y) |_{r = r'}
  \end{pmatrix}\] 
is invertible. By the same argument as in the beginning of the proof, it is equivalent to showing the following linear map 
\[ V_{\hbar_N,E} \to \R^2, \quad \phi \mapsto (\phi, \dell_r \phi)|_{(r,\w)} \] 
is surjective for any $(r,\omega) \in \R^d \RM \{0\}$. We will provide two functions $\phi^{(1)}, \phi^{(2)}$ in $V_{\hbar_N,E}$, whose span  surjects to $\R^2$. Without loss of generality, we may assume $\omega$ is in the positive $x_1$ direction, then $\pa_r \phi=\pa_{x_1}\phi$. If $N$ is even, we take
\bea
\phi^{(1)}(r,\omega) &=& \phi_{(0,N,0,\cdots,0),\hbar}(r,0,\cdots,0) = c_1 e^{-r^2/2h} \\
\phi^{(2)}(r,\omega) &=& \phi_{(2,N-2,0,\cdots,0),\hbar}(r,0,\cdots,0) = c_2 (r^2/h-1) e^{-r^2/2h}, 
\eea
and if $N$ is odd, we take
\bea
\phi^{(1)}(r,\omega) &=& \phi_{(1,N-1,0,\cdots,0),\hbar}(r,0,\cdots,0) = c_1 ((h^{-1/2} r) e^{-r^2/2h} \\
\phi^{(2)}(r,\omega) &=& \phi_{(3,N-3,0,\cdots,0),\hbar}(r,0,\cdots,0) = c_2 ((h^{-1/2} r)^3 - 3 (h^{-1/2} r)) e^{-r^2/2h}, 
\eea
where we used that the $1$D Hermite functions satisfy $\phi_{\ell, \hbar}(0)\neq 0$ if $\ell$ is even, and $c_1, c_2$ are $r$-independent constant. Rescaling the variable $r \mapsto r \hbar^{1/2}$ to get rid of the $\hbar$ dependence and omitting the $c_i$ factors, we may verify that the image of $\phi^{(1)}, \phi^{(2)}$ are independent: 
\[ \det \left( 
\begin{array}{cc} 
\phi^{(1)} & \pa_r (\phi^{(1)}) )\\
\phi^{(2)} & \pa_r (\phi^{(2)}) ) \end{array} \right) = \det \left( 
\begin{array}{cc} 
f e^{-r^2/2} & \pa_r (f e^{-r^2/2} )\\
g e^{-r^2/2} & \pa_r (g e^{-r^2/2}) \end{array} \right)
= e^{-r^2} \det \left( 
\begin{array}{cc} 
f & \pa_r f  \\
g  & \pa_r g \\  
  \end{array} \right)\]
where $N$ is even, $\{f,g\} = \{1, r^2-1\}$, if $N$ is odd, $\{f,g\}=\{r, r^3-3r\}$. A direct computation shows that the above determinant is non-zero. Hence $\Cov\lr{\Phi_{\hbar, E}, \dell_r \Phi_{\hbar, E}}$ is invertible, and the covariance matrix in Eq \eqref{E:1jet} is non-degenerate. This finishes the proof for the proposition. 
\end{proof}

We then have,

\begin{lem}[Kac-Rice for Gaussian Fields]\label{L:Gaussian KR}
Let $\Phi_{\hbar, E}$ be the random Hermite eigenfunction of $\widehat{H}_\hbar$ with eigenvalue $E$ as in \eqref{PHIN}. Then the density of zeros of $\Phi_{\hbar, E}$ is given by
\begin{equation}
F_{\hbar, E}(x)= \lr{2\pi}^{-\frac{d+1}{2}}\int_{\R^d}|\W^{1/2}(x)\xi| \;\; e^{-\abs{\xi}^2/2} \;\; d\xi,\label{E:Gaussian KR}
\end{equation}
where $\W(x)$ is the $d\x d$ matrix
\begin{align}
\notag \W_{ij}(x) &= (\dell_{x_i}\dell_{y_j} \log \Pi_{\hbar, E})(x,x)\\
\label{E:Gaussian Cov Mat} &=
\frac{(\Pi_{\hbar, E}\cdot \dell_{x_i}\dell_{y_j}\Pi_{\hbar, E})(x,x)-(\dell_{x_i}\Pi_{\hbar, E} \cdot \dell_{y_j}\Pi_{\hbar, E})(x,x)}{\Pi_{\hbar, E} (x,x)^2}
\end{align}
and $\Pi_{\hbar, E}(x,y)$ is the kernel of eigenspace projection \eqref{COV}.
\end{lem}

We refer to \cite{HZZ} for background. The main task in proving results on zeros near the caustic is therefore to work out the asymptotics of $\Pi_{\hbar, E}(x,x)$ and its derivatives there.

\section{The tube region for $\alpha = \frac{2}{3}$: Proofs of Theorems \ref{SCLintro} and  \ref{CAUSTIC}}\label{TUBESECT}

This section is the heart of the paper, in which we determine the Airy scaling asymptotics of the eigenspace projections \eqref{COV} and of the Kac-Rice matrix \eqref{E:Gaussian Cov Mat}  in a $\hbar^{2/3}$-tube around the caustic. We also find the scaling asymptotics of the derivatives of the kernel and prove the Kac-Rice formula of Theorem \ref{CAUSTIC}. We begin with an outline of the proof of Theorem \ref{SCLintro} (and Theorem \ref{T:Mainb} since its proof follows a similar pattern). 

 % In the next Proposition  we derive Airy scaling asymptotics 
%for the scale eigenspace projection kernels, and then   derive scaling %asymptotics of the Kac-Rice matrix.
% For simplicity of exposition,  we put some technical details in the Appendix %(Section \ref{APPB}).  Theorem \ref{CAUSTIC}  follows from these scaling %asymptotics and the Kac-Rice formula of Theorem  \ref{CAUSTIC}.

%Our proof of Theorem \ref{SCLintros} proceeds by a steepest descent analysis of the expressions \eqref{eq:Pi_Contour}. It is divided schematically in three steps:
%\begin{enumerate}
%\item[(i)] Deform the contour to pass through the (complex) critical point along the steepest descent direction;
%\item[(ii)] Excise away the part of the contour sufficiently far away from the critical points (see Lemma \ref{lm:excision}); 
%\item[(iii)] Calculate the contribution of the contour integral near the critical point. 
%\end{enumerate}

\subsection{ Outline of the proof of Theorems \ref{SCLintro} and \ref{T:Mainb}}\label{S:Outline}
The proofs of Theorems \ref{SCLintro} and \ref{T:Mainb} are based on a steepest descent analysis of the contour integral \eqref{eq:pixy} for the eigenspace projection kernel. They involve a number of tricky technical steps, which we now sketch. 

We fix $E=1/2$ (i.e. set $\hbar=\hbar_N=\lr{2N+d}^{-1}$) and drop the $E$ subscript. Using \eqref{eq:pixy}, the spectral projector $\Pi_\hbar$ and the derivatives appearing in the Kac-Rice matrix $\W$ (given in \eqref{E:Gaussian Cov Mat}) can be written as
\begin{align}
\Pi_{\hbar}(x,x)  &=  \oint_{C_\epsilon} A(z) e^{\Phi(z)/\hbar} dz \label{eq:Pi_Contour}\\
\pa_{x_i} \Pi_{\hbar}(x,x) &=  \oint_{C_\epsilon } \left( -\frac{x_i}{\hbar} \frac{1-z}{1+z} \right) A(z) e^{\Phi(z)/\hbar} dz  \notag\\
\pa_{x_i} \pa_{y_j} \Pi_{\hbar}(x,x) &=  \oint_{C_\epsilon}  \left[  \frac{x_i x_j}{\hbar^2} \left(\frac{1-z}{1+z} \right)^2  +  \frac{\delta_{ij}}{\hbar} \frac{2z}{1-z^2} \right]  A(z) e^{\Phi(z)/\hbar} dz \notag.
\end{align}
The phase function $\Phi(z)$ and the amplitude $A(z)$ are 
\be
\label{eq:PhiandA}
\Phi(z) = -\frac{1-z}{1+z} \; |x|^2  - \frac{1}{2} \log z ,\quad A(z) dz  = \left(\frac{z}{\pi \hbar (1-z^2)} \right)^{d/2} \frac{ dz }{2\pi i z}.
\ee

Note that the integrand $A(z)e^{\Phi(z)/\hbar}$ is defined on 
\begin{equation}
S=\C ~\backslash ~\lr{(\infty, -1]\cup [1,\infty)\cup \set{0}}.\label{E:SDef}
\end{equation}
Indeed, the term $e^{-\log(z)/(2\hbar)}$ from $e^{\Phi(z)/\hbar}$ and $z^{d/2-1}$ from $A(z)$ combine to give an integer total power of $z$ 
\[z^{\frac{d}{2}-\frac{1}{2\hbar}-1}=z^{-N-1},\]
 where $N=\frac{1}{2}\lr{\hbar^{-1}-d}$ is the degree of the Hermite functions in the $\frac{1}{2}-$eigenspace of $\widehat{H}_\hbar.$  Observe that the integrand has singularities at $z = \pm 1$ and the critical points 
\be \label{CPE} z_\pm =1 \pm 2 |x| \sqrt{|x|^2-1} + 2 (|x|^2-1) \ee
of $\Phi(z)$ do not lie on $C_\epsilon.$ If $x$ is in the $\hbar^\alpha$ neighborhood of the caustic $\ccal_E$ (the unit circle), say 
\begin{equation}
|x|^2 = 1 + \hbar^\alpha s,\label{E:Caustic Scaling}
\end{equation}
then
\[ z_\pm = 1 \pm 2\hbar^{\alpha/2} \sqrt{s} + O(\hbar^\alpha), \]
and the Hessian of $\Phi(z)$ at the critical points is 
\[ \Phi''(z_\pm) = \mp 2^{-3/2}\hbar^{\alpha/2} \sqrt{s} +  O(\hbar^\alpha).\]
The proofs of Theorems \ref{SCLintro} and \ref{T:Mainb} proceed schematically in three steps:
\begin{enumerate}
\item[(i)]  Deform the contour such that it passes through the critical points of the phase function $\Phi(z)$, and wraps around the singular points of the amplitude function $A(z)$, 
\item[(ii)]  Show that the contribution from the contour away from the critical points and singular points are irrelevant. 
\item[(iii)] Calculate the leading term contribution from the contour near the critical points and the singular points. 
\end{enumerate}

\begin{figure}
   \centering   
     \includegraphics[width=0.3\textwidth]{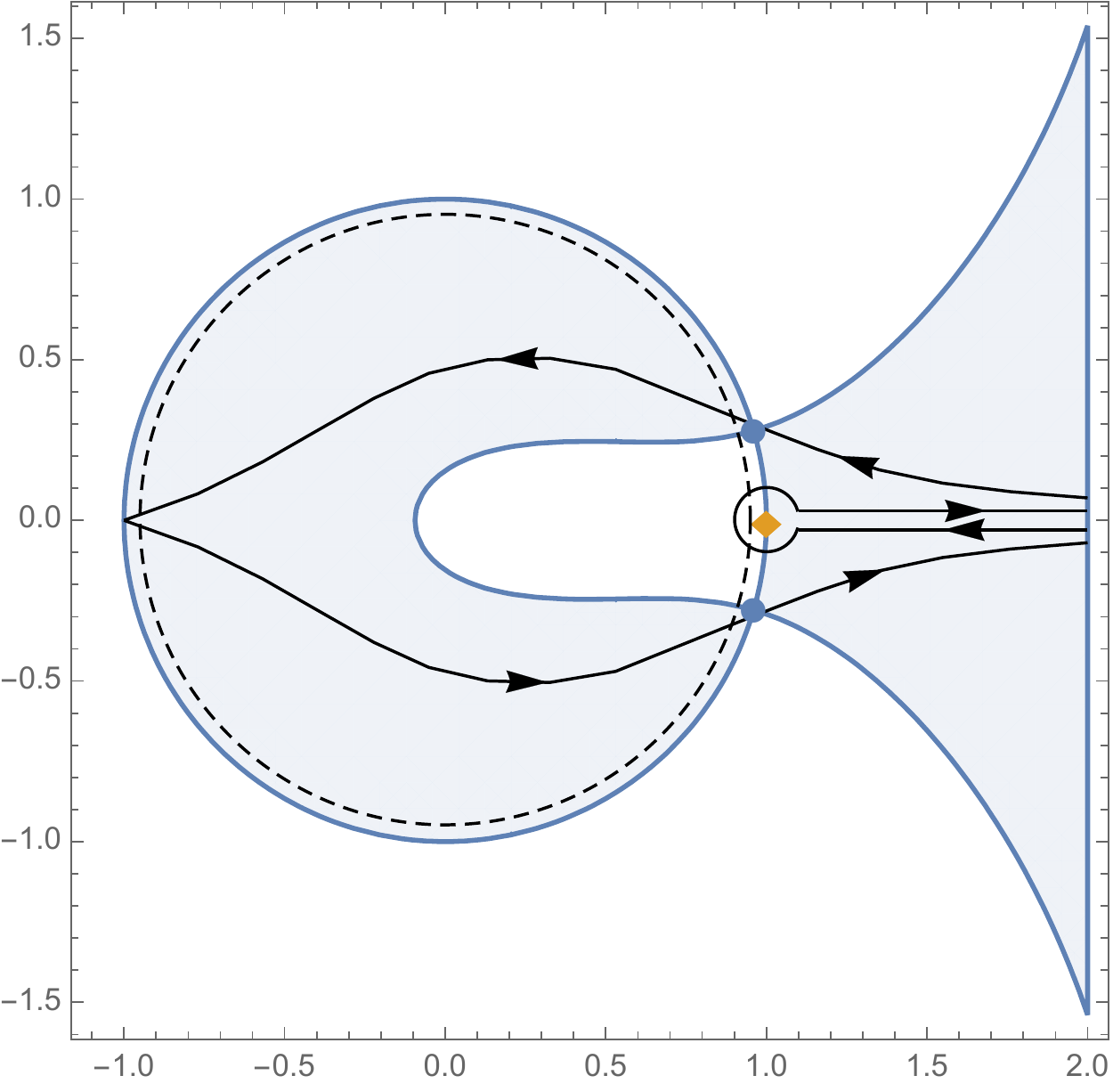}
       \includegraphics[width=0.3\textwidth]{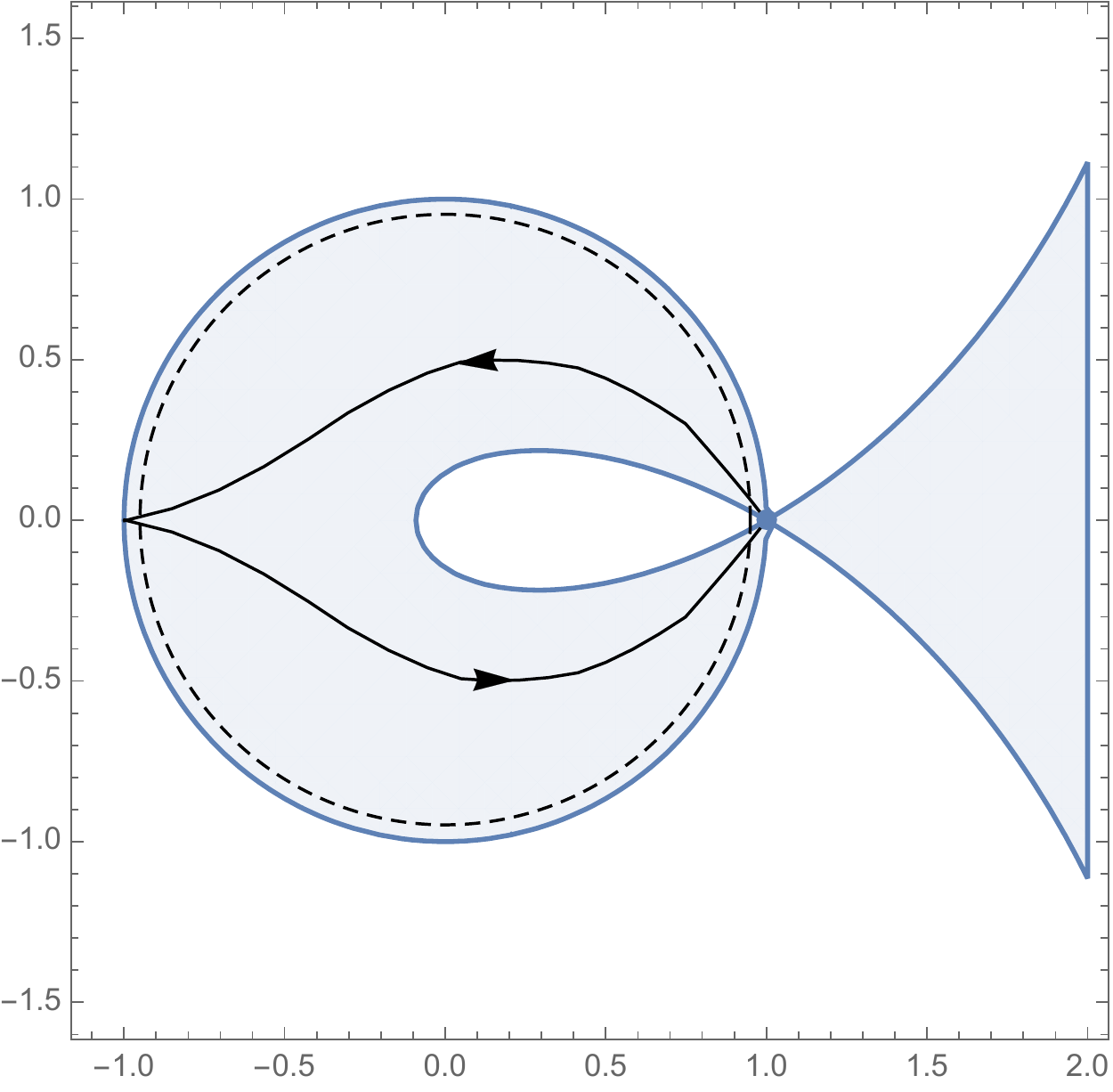}
        \includegraphics[width=0.3\textwidth]{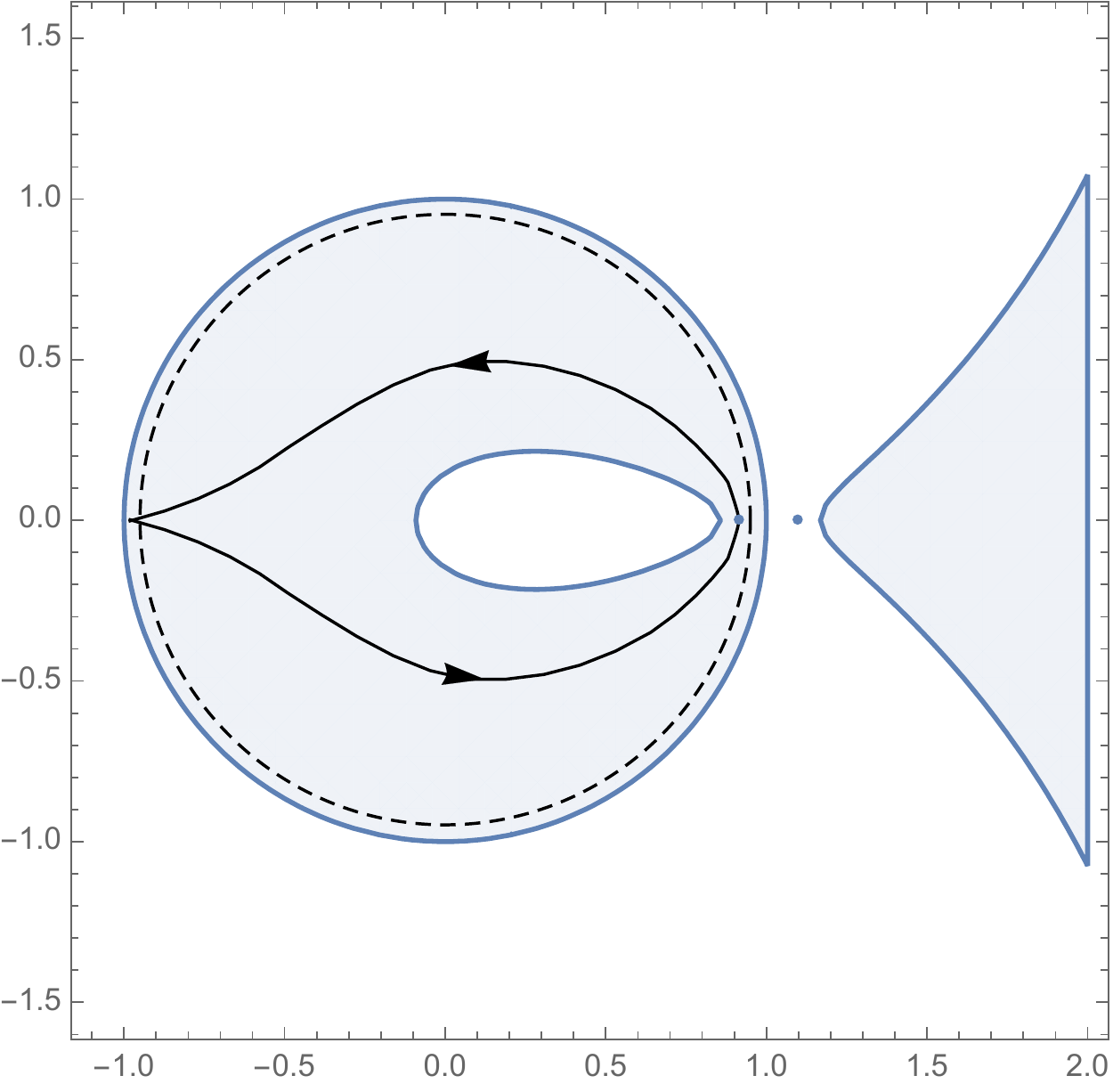}
    \caption{The contour deformation for  $|x|^2 < 1,  |x|^2=1$ and $|x|^2 >1$. Regions where $\text{Re}(\Phi)<0$ are shaded blue.\label{fig:safe-region} }  
\end{figure}

The contour deformations we use in step (i) are shown in Figure \ref{fig:safe-region}, where the panels from left to right correspond to the allowed region ($x\in \acal_E$ and $\alpha <2/3$), the caustic region ($\alpha=2/3$) and the forbidden region ($x\in \fcal_E$ and $\alpha<2/3$) respectively. The black oriented lines are the deformed contours, and the dashed line corresponds to the original contour. The blue regions are those where $\text{Re}(\Phi)<0$ and hence the integrand is rapidly decaying in $\hbar.$ Note that 
\begin{enumerate}
\item  In the allowed region, the critical points lie on the unit circle; in the caustic region $|x|^2=1$, the two critical point merge at $z=1$; in the forbidden region, the two critical points are real, lie on the opposite sides of $1$, and the deformed contour passes through the critical point inside the circle. 
\item In all three cases, the contour goes to the point $z=-1$ with finite slope, and $\Re \Phi(z) \to -\infty$, in agreement with the contour following the downward gradient flow of $\Re \Phi(z)$. 
\item In the allowed region case (left panel), there is an additional `key-hole' contour, caused by the singularity of $A(z)$ at $z=1$. 
\end{enumerate}

There are two non-standard aspects in our stationary phase integral when $|x|\approx 1$, which are particularly important for the proof of Theorem \ref{T:Mainb}. The first is the coalescing of the critical points when $\abs{x}=1,$ which causes the usual stationary phase method with quadratic phase function to break down. The second is the singularity in the amplitude $A(z)$ at $z=\pm 1$. The singularity at $z=-1$ is less problematic, since the phase function $\Re(\Phi(z)) \to -\infty$ as $z \to -1$ along the steepest descent path, making the integral $e^{\Phi(z)/\hbar} A(z)$ convergent. The singularity of $A(z)$ at $z = 1$, however, coincides with the critical point $z_c = 1$ if $|x|^2=1$ right on the caustic. 

A key point in the proof of Theorem \ref{SCLintro} (when $\abs{x}^2=1+\hbar^{2/3} s$) is the form of the deformed contour $C_\epsilon$ in the central panel of Figure \ref{fig:safe-region}. Namely, in an $\hbar^{1/3}$ neighborhood of $z=1$ the deformed contour becomes (up a sign change) to the Airy contour $\mathcal C$ (see the left panel in Figure \ref{fig:airy-contour}). The integral over this portion of the contour will turn out to give the leading order behavior of $\Pi_\hbar(x,x).$ The phase function $\Phi,$ when properly rescaled, becomes the Airy phase and it is in this way that the primitives and derivatives $\Ai_k$ of the Airy function appear in the statement of Theorem \ref{SCLintro} (see Section \ref{S:Outline of Proofs} as well as the beginning of the proof of Proposition \ref{pp:caustic} for more on this point).

In the proofs of both Theorem \ref{SCLintro} and \ref{T:Mainb} we implement step (ii) for each critical point $z_c$ by using the Excision Lemma \ref{lm:excision} to localize, to leading order, the integral \eqref{eq:Pi_Contour} (with the contour deformed as above) to a neighborhood of $z_c$
\[ \Gamma_{z_c, \delta}: = \{z \in \Gamma \mid 0 > [\Re(\Phi(z)) - \Re(\Phi(z_c))]/\hbar > -\hbar^{-\delta} \} \]
for some $\delta>0$.  
%For more detail, see the Lemma \ref{lm:excision} in Appendix \ref{APPB}.   

Finally, in step (iii), to get the leading order order behavior of \eqref{eq:Pi_Contour}, we evaluate the localized integral over $\Gamma_{z_c,\delta}$ by Taylor expanding $A(z)$ and $\Phi(z)$ around $z_c$ and use the Trimming Lemma \ref{lm:trim} to estimate the error term.

\subsection{Two Technical Lemmas}

As discussed above, after deforming the contour to good position, we simplify the integrals in two ways. Since the simplifications are a bit technical, we prove the relevant Lemmas before getting into the proof of Theorem \ref{SCLintro}. The two Lemmas are general statements about oscillatory integrals with complex phase.

\begin{lem}[Excision Lemma]\label{lm:excision}
  Let
  \[ I_\hbar = \int_{\Gamma_\hbar} f_\hbar(z) e^{\Phi_\hbar(z)} dz \]
  where for each $\hbar \in (0, \hbar_0)$ we have the following conditions satisfied \\
  (1) $\Gamma_\hbar$ is a compact smooth curve in $\C$, and there are $C_1, \alpha>0$, such that $\text{Length}(\Gamma_\hbar) < C_1 \hbar^{-\alpha}$ \\
  (2) $f_\hbar(z)$ is analytic in a neighborhood of $\Gamma_\hbar$, and
  there are $C_2, \beta>0$, such that
  \[ \sup_{\hbar \in (0, \hbar_0), z \in \Gamma_\hbar} |f_\hbar(z)|
  \hbar^\beta < C_2\]
  (3) $\Phi_\hbar(z)$ is analytic in a neighborhood of $\Gamma_\hbar$,
  and there are constants $C_3, \gamma> 0$, such that
  \[ \sup_{\hbar \in (0, \hbar_0), z \in \Gamma_\hbar} Re(\Phi_\hbar(z))
  \hbar^\gamma < -C_3 \] Then
  \[ I_\hbar = O(\hbar^\infty) \]
\end{lem}
\bpf
For all $\hbar \in (0, \hbar_0)$, we have
\[ |I_\hbar| \leq \int_{\Gamma_\hbar} |f_\hbar(z)| e^{Re(\Phi_\hbar(z))} dz 
\leq  C_2 \hbar^{-\beta} C_1 \hbar^{-\alpha} e^{-C_3 \hbar^{-\gamma}}  
= O(\hbar^\infty)\]
\epf

\bl[Trimming Lemma]
\label{lm:trim}
Fix $a,C,m>0$ and $\epsilon\in \lr{0,a/m}.$ Let $\Gamma$ be a contour in $\C$ that goes to infinity along $\arg(z)=\theta_{\pm}$,  $f(z)$  a holomorphic function in a neighborhood of $\Gamma$, and let  $A_\hbar(z),B_{\hbar}(z)$ be smooth functions in a neighborhood of $\Gamma$. Assume that
\[|f(z)| < C (1+|z|^m),\quad  |A_\hbar(z)| + |B_\hbar(z)| \leq C \hbar^a (1 + |z|^m), \quad \forall z \in  \Gamma \cap D(\hbar^{-\epsilon}),\]
where $D(r)$ is the disk of radius $r$ centered at the origin. For each $\hbar>0,$ let
\[ I_{\hbar} = \int_{\Gamma \cap D(\hbar^{-\epsilon}) } (f(z) + B_\hbar(z)) e^{\Phi(z)+ A_\hbar(z)} dz, \quad \Phi(z) = \sum_{k=0}^d a_k z^k,\]
with  $\Re(a_d e^{id\theta_\pm}) < 0$. Then, as $\hbar\gives 0,$ we have
\[ I_{\hbar} = \int_\Gamma f(z) e^{\Phi(z)} dz + O(\hbar^a).\]
\el

\bpf
For $\hbar$ small enough such that $2 C \hbar^{a - m \epsilon} < 1$, we have
\[ |A_\hbar(z)| + |B_\hbar(z)| \leq C \hbar^a (1 + |z|^m) \leq C (\hbar^a + \hbar^{a-m \epsilon})  \leq 1, \quad e^{|A_\hbar(z)|} \leq e \quad \forall z \in  \Gamma \cap D(\hbar^{-\epsilon}).\]
Let $I_0 =  \int_\Gamma f(z) e^{\Phi(z)} dz $. By assumption,  $\Re(\Phi(z)) < 0$ as $z \to \infty$ along $\Gamma$ and $f(z)$ has only polynomial growth, hence $I_0$ is finite. Next, define 
\[I_{\hbar,s} = \int_{\Gamma \cap D(\hbar^{-\epsilon}) } (f(z) + s B_\hbar(z)) e^{\Phi(z)+ s A_\hbar(z)} dz, \quad s \in [0,1].\]
From Taylor's formula, we have
\[ |I_{\hbar, 1} - I_{\hbar, 0}| \leq \sup_{s \in [0,1]} \left| \frac{d}{ds} I_{\hbar, s} \right|, \]
and there exists $C'>0$ so that 
\bea 
\left| \frac{d}{ds} I_{\hbar, s} \right|&=& \left| \int_{\Gamma \cap  D(\hbar^{-\epsilon})} [B_\hbar(z) + (f(z) + s B_\hbar(z))A_\hbar(z)]e^{\Phi(z)+ s A_\hbar(z)} dz \right| \\
& \leq &  \int_{\Gamma \cap D(\hbar^{-\epsilon})} C' \hbar^a (1+|z|^{2m}) e^{|A_\hbar(z)|} e^{\Re(\Phi(z))}  |dz| \\
& \leq & eC' \hbar^a \int_{\Gamma} (1+|z|^{2m}) e^{\Re(\Phi(z))} |dz| \leq C'e \hbar^a
\eea
Finally, we note that for some $C''>0$
\[ | I_0 - I_{\hbar, 0}|  = |\int_{\Gamma \RM D(\hbar^{-\epsilon})} f(z) e^{\Phi(z)} dz| = O(e^{-C'' \hbar^{-d \epsilon}+\epsilon})  = O(\hbar^\infty) \]
This shows that $|I_\hbar - I_0| =O(\hbar^a)$ and completes the proof.
\epf

With these preparations, we are now ready to prove Theorem \ref{SCLintro}.

\subsection{Diagonal scaling asymptotics of Theorem \ref{SCLintro}}
The proof of the scaling asymptotics \eqref{E:CausticScaling} of the eigenspace projection kernels is almost the same along the diagonal $u=v$ as it is for $u\neq v.$ Since the on-diagonal result
\begin{equation}
\Pi_{\hbar,1/2} (x_0 + \hbar^{2/3} u, x_0 + \hbar^{2/3} u) = \hbar^{-2d/3+1/3} \Pi_0(u, u) (1 + O(\hbar^{1/3})) \label{E:CausticScalingdiag}
\end{equation}
and its analogs for the for the derivatives $\pa_{x_i} \Pi_{\hbar}(x,x)$,  $\pa_{x_i} \pa_{y_j} \Pi_{\hbar}(x,x)$ are the key ingredients in obtaining the scaled Kac-Rice formulae of Theorem \ref{CAUSTIC}, we state them separately in the following Proposition. We then indicate in Section \ref{S:CompletionPf} the additional steps required to obtain the off-diagonal statement of Theorem \ref{SCLintro}.

\bp
\label{pp:caustic}
Let $\abs{x}^2 = 1 + \hbar^{2/3} s$ for any $s \in \R$, then
\be\label{pi-tube} \Pi_{\hbar}(x,x) = 2^{-d+1}\pi^{-d/2} \hbar^{(1-2d)/3} \Ai_{-d/2}(s)(1+O(\hbar^{1/3})). \ee
Moreover, the entries of the Kac-Rice matrix \eqref{E:Gaussian KR} 
have the scaling asymptotics,
\be \label{omega-tube} \Omega_{ij} = \hbar^{-4/3} \left\{ x_i x_j \left[ \frac{\Ai_{2-d/2}(s) }{\Ai_{-d/2}(s)} - \left(\frac{\Ai_{1-d/2}(s) }{\Ai_{-d/2}(s)} \right)^2\right] + 2^{-1} \delta_{ij}  \frac{\Ai_{-1-d/2}(s) }{\Ai_{-d/2}(s)}   \right\} (1+O(\hbar^{1/3})).
\ee
The implied constants are uniform when $s$ varies over a compact set.
\ep

\bpf
First, we use the integral expression \eqref{eq:Pi_Contour} to write $\Pi_\hbar(x,x) = \int_{C_\epsilon} A(z) e^{ \Phi(z)/\hbar}$, where $\Phi(z)$ and $A(z)$ are given in \eqref{eq:PhiandA} and $C_\epsilon=\set{\abs{z}=1-\epsilon}$ counter clockwise. 

%Near $z=1$, if we introduce coordinate $T$ such that $z = 1 - \hbar^{1/3} T$, we may get
%\[ A(z) dz = (2\pi)^{-d/2} \hbar^{1/3-2d/3} T^{-d/2} (1+O(\hbar^{1/3} T)) \frac{-dT}{2\pi i}, \quad \Phi(z) = \]

Then, we deform the contour $C_\epsilon$ along the steepest descent path, except bending it a bit near $z=1$ singularity. More precisely, fix  $1 \gg \delta > 0$, we  then define 
\[ z'_\pm = 1 - \hbar^{1/3-\delta} e^{\pm i \pi /3} \] 
and denote $\Gamma'_\pm$ to be the downward gradient flowlines of $\Re \Phi(z)$ starting from $z'_\pm$ ending at $-1$, as shown in the middle panel of Figure \ref{fig:safe-region}. Then we define the deformed contour  $\Gamma := \Gamma_1 \cup \Gamma_2 \cup \Gamma_3$, where \\
\[ \begin{cases}
 \Gamma_1 = (\cup_{a=\pm} \Gamma'_a) \cap \{ |z+1| < \delta \} \\
 \Gamma_2 = (\cup_{a=\pm} \Gamma'_a) \cap \{ |z+1| > \delta \}\\
 \Gamma_3 = \{ 1 - r \hbar^{1/3} e^{i \pi /3} \mid r \in [1, \hbar^{-\delta}]\} \cup \{ 1 - r \hbar^{1/3} e^{-i \pi /3} \mid r \in [1, \hbar^{-\delta}]\} \\
\qquad\quad  \cup \{ 1 - \hbar^{1/3} e^{i \theta} \mid \theta \in [-\pi/3, \pi/3]\}
\end{cases}
\]
Next, we estimate the contribution from the three parts of the contours. 

On $\Gamma_1$, the amplitude satisfies
\[ \abs{A(z)}\leq C \hbar^{-d/2}|z+1|^{-d/2}, \]
while the real part of the phase obeys
\[ \Re(\Phi(z)/\hbar)\leq - C (\hbar|z+1|)^{-1}.\]
By the Excision Lemma \ref{lm:excision},  the integral over $\Gamma_1$ is $O(\hbar^\infty)$. 

On $\Gamma_2$, we have
\[ \sup_{z \in \Gamma_2} \Re \Phi(z) = \Re \Phi( z'_\pm) = -\frac{1}{24} \hbar^{-3\delta} (1 + O(\hbar^{2\delta})) \]
with the implied constants uniform when $s$ varies over compact sets. Therefore, by the Excision Lemma  \ref{lm:excision} again, the integral over $\Gamma_2$ is $O(\hbar^\infty)$. 

Finally, on $\Gamma_3$, we note that $\hbar^{1/3} < |z-1| < \hbar^{1/3-\delta}$, hence over $\Gamma_3$ for small enough $\hbar$, we have
\[ A(z) = (2 \pi \hbar (1-z))^{-d/2} (1 + R_1(z)) \] 
where there exists $C_1>0$ such that $|R_1(z)| < C_1 |z-1|$  for all $|z-1| < \hbar^{1/3-\delta}$. Hence over $\Gamma_3$, we have
\[ |(1-z)^{-d/2} R_1(z)| < C_1 |z-1|^{-d/2+1} < C_1 \hbar^{(-d/2+1)(1/3)} \] 
Similarly, for $\Phi(z)$, we consider its Taylor expansion around $z=1$: 
\[ \Phi(z) = \frac{\hbar^{2/3} s }{2}(z-1) + \underbrace{\frac{\hbar^{2/3} s }{4}(z-1)^2}_{R_3(z)} + \left(-\frac{(z-1)^3}{24} + \underbrace{\frac{\hbar^{2/3} s (z-1)^3}{8}}_{R_4(z)} \right) +  R_2(z) \]
where there exists $C_2>0$, which is uniformly bounded when $u$ varies over a compact set, such that $|R_2(z)| < C_2 |z-1|^4$  for all $|z-1| < \hbar^{1/3-\delta}$. We may rewrite the integral using $z = 1 - \hbar^{1/3} T$ and reverse the contour orientation, then 
\[ T \in \Gamma_T := \{ r e^{\pm i \pi /3} \mid r \in [1, \hbar^{-\delta}]\} \cup \{ e^{i \theta /3} \mid \theta \in [-\pi/3, \pi/3] \} \]
The phase function and the amplitude become
\[ - A(z(T)) dz(T) = (2\pi)^{-d/2} \hbar^{1/3-2d/3} [T^{-d/2} + \tilde{A}_\hbar (T)] \frac{dT}{2\pi i}\]
where from the bound on $R_1(z)$ we have
\[ |\tilde{A}_\hbar (T)| = |T^{-d/2} R_1(z(T))| \leq C_1 |T|^{-d/2+1} \hbar^{1/3} < C_1 \hbar^{1/3} \]
and 
\[ \hbar^{-1} \Phi(z(T)) = \frac{T^3}{24} - \frac{T s}{2} + \hbar^{-1}(R_2(z(T)) + R_3(z(T)) + R_4(z(T)) ).\]
The remainder can be bounded as follows
\bea
\hbar^{-1}|R_2(z(T))| & \leq &  C_2 \hbar^{1/3} T^4 \\
\hbar^{-1}|R_3(z(T))| & \leq &  (|s|/4) \hbar^{1/3} T^2 \\
\hbar^{-1}|R_4(z(T))| & \leq &  (|s|/8) \hbar^{2/3} T^3.
\eea
The above argument shows that we may apply the Trimming Lemma \ref{lm:trim}, to get the contribution from $\Gamma_3$ to be
\bea 
I_3 &=& \int_{\Gamma_3} A(z) e^{\Phi(z)/\hbar} dz  \\
 & = & (2\pi)^{-d/2}\hbar^{(1-2d)/3} \int_{\ccal} T^{-d/2} e^{\frac{T^3}{24} - \frac{T s}{2}} \frac{dT}{2\pi i} (1 + O(\hbar^{1/3}) \\
 & = & 2^{-d+1}\pi^{-d/2} \hbar^{(1-2d)/3} \Ai_{-d/2}(s)(1+O(\hbar^{1/3}))
\eea
This yields the desired result for $\Pi_\hbar(x,x)$. The same argument applies straightforwardly to $\pa_{x_i} \Pi_{\hbar}(x,x)$,  $\pa_{x_i} \pa_{y_j} \Pi_{\hbar}(x,x)$. We get
\bee
\Pi_{\hbar}^{-1}\pa_{x_i} \Pi_{\hbar}(x,x)  &=& \frac{ \oint  \left( -\frac{x_i}{\hbar} \frac{1-z}{1+z} \right) A(z) e^{\Phi(z)/\hbar} dz }{ \oint   A(z) e^{\Phi(z)/\hbar} dz } \notag\\
&=&  - x_i \hbar^{-2/3} \frac{\Ai_{-d/2+1}(s)}{\Ai_{-d/2}(s)}(1+O(\hbar^{1/3}))\label{1st-der}
\eee
and
\bee
\Pi_{\hbar}^{-1} \pa_{x_i} \pa_{y_j} \Pi_{\hbar}(x,x) &=&  \frac{\oint  \left[  \frac{x_i x_j}{\hbar^2} \left(\frac{1-z}{1+z} \right)^2  +  \frac{\delta_{ij}}{\hbar} \frac{2z}{1-z^2} \right]  A(z) e^{\Phi(z)/\hbar} dz }{\oint   A(z) e^{\Phi(z)/\hbar} dz} \notag\\
&=& \left(x_i x_j \hbar^{-4/3} \frac{\Ai_{-d/2+2}(s)}{\Ai_{-d/2}(s)} + \delta_{ij} \hbar^{-4/3} \frac{\Ai_{-d/2-1}(s)}{2\Ai_{-d/2}(s)}\right)(1+O(\hbar^{1/3})). \label{2nd-der}
\eee
Combining these, we get
\bea
\Omega_{ij}(u) &=& 
\Pi_{\hbar}^{-1} \pa_{x_i} \pa_{y_j} \Pi_{\hbar}(x,x) -\Pi_{\hbar}^{-2}\pa_{x_i} \Pi_{\hbar}(x,x) \pa_{x_j} \Pi_{\hbar}(x,x) \\
&=& \hbar^{-4/3} \left(\delta_{ij} \frac{\Ai_{-d/2-1}(s)}{2\Ai_{-d/2}(s)} + x_i x_j \left( \frac{\Ai_{-d/2+2}(s)}{\Ai_{-d/2}(s)} -\frac{\Ai^2_{-d/2+1}(s)}{\Ai^2_{-d/2}(s)} \right) \right)(1+O(\hbar^{1/3})).
\eea
This completes the proof of Proposition \ref{pp:caustic}.
\epf

\begin{remark}
By substituting $s \mapsto \hbar^{\alpha-2/3} s$, and applying the asymptotic expansion of the weighted airy function $Ai_{k}$ (see Proposition \ref{ppBPD}), we may recover the leading term of the corresponding results on $\Pi$ and $\Omega$ in the allowed and forbidden annuli. However, this does not give an estimate of the error terms. We leave the more detailed analysis for Section \ref{ANNULISECT}.
\end{remark}

\subsection{Proof of  the off-diagonal scaling asymptotics of  Theorem \ref{SCLintro}}\label{S:CompletionPf}
To complete the proof of Theorem \ref{SCLintro}, we now give the full off-diagonal scaling asymptotics of the covariance function \eqref{COV}. The proof is a development of the diagonal result in Proposition \ref{pp:caustic}. We fix $x_0 \in \ccal_E$ and consider
\begin{equation} \label{SC} \Pi_{\hbar, E}(x_0 + \hbar^{2/3} u, x_0 +  \hbar^{2/3}v) = \sum_{\beta \in (\Z_{\geq 0})^d, |\beta|=N} \varphi_{\hbar, \beta}(x_0 + \hbar^{3/2} u)
\varphi_{\hbar,\beta}(x_0 + \hbar^{3/2} v). \end{equation}
Theorem \ref{SCLintro} asserts that (up to a scalar factor) the scaling limit of the kernels \eqref{SC} is the kernel \eqref{Pi0}.

\begin{proof} The proof of the off-diagonal scaling asymptototics
is similar to that of the on-diagonal, so we only give a brief sketch of it.

Repeating the proof of Proposition \ref{pp:caustic}, we again localize the integral for $\Pi_\hbar(x,x)$ to the contour $\Gamma_3.$ We then rescale $z =1- \hbar^{1/3} T$ in \eqref{eq:pixy} and again apply the Trimming Lemma \ref{lm:trim} to obtain 
\be 
\label{eq:piuv}
\Pi_{\hbar} (x_0 + \hbar^{2/3} u, x_0 +\hbar^{2/3} v)  = (2\pi)^{-d/2}  \hbar^{-2d/3+1/3} \int_\ccal T^{-d/2} e^{- \frac{(u-v)^2}{2T} + \frac{T^3}{24} - \frac{T}{2} \lan u+v, x_0 \ran}  dT(1+ O(\hbar^{1/3})).
\ee

On the other hand, we may rewrite \eqref{Pi0} using Lemma \ref{airy-double} on products of Airy functions. We change the integration variable $T \mapsto 2^{-4/3} T$ to get,
\[\Ai(2^{1/3}(u_1 + p^2/2))\Ai(2^{1/3}(v_1 + p^2/2)) = \int_\ccal e^{ \frac{ T^3}{24} - (u_1+v_1+p^2) \frac{T}{2} - \frac{ (u_1 - v_1)^2}{2 T} } \sqrt{\frac{1}{2 \pi T}} \frac{ 2^{-2/3} d T}{2 \pi i} 
\] Substituting into  \eqref{Pi0}, and computing the $dp$ Gaussian integral, we get

\begin{align}
\notag \Pi_0(u_1, u'; v_1, v') &= (2 \pi)^{(-d+1)/2}  \int_\ccal T^{-(d-1)/2} e^{ \frac{ T^3}{24} - (u_1+v_1+p^2) T/2 - \frac{ (u_1 - v_1)^2}{2 T}  - \frac{(u'-v')^2}{2 T} } \sqrt{\frac{1}{2 \pi T}} \frac{  d T}{2 \pi i} \\
&= (2\pi)^{-d/2}  \int_\ccal T^{-d/2} e^{ \frac{ T^3}{24} - (u_1+v_1)\frac{T}{2} - \frac{ (u - v)^2}{2 T}}  \frac{  d T}{2 \pi i} \label{eq:Pifin}
\end{align}
This agrees with the right hand side of  \eqref{eq:piuv}, hence gives the proof for  off-diagonal formula \eqref{E:CausticScaling}. 
This combines with the diagonal result \eqref{pi-tube} in Proposition \ref{pp:caustic} finishes the proof of Theorem \ref{SCLintro}.
\epf

\begin{remark}
We may also easily derive the result of $\Omega_{ij}$ in the $\alpha=2/3$ region from the off-diagonal scaling limit of $\Pi_\hbar$. 
\bea
\Omega_{ij}(x_0+\hbar^{2/3} u_0) &=& \hbar^{-4/3}  \pa_{u_i} \pa_{v_j}|_{u=v=u_0} \log (\Pi_\hbar (x_0 + \hbar^{2/3} u, x_0 + \hbar^{2/3} v)) \\
&=& \hbar^{-4/3}  \pa_{u_i} \pa_{v_j}|_{u=v=u_0} \log (\Pi_0(u,v)) (1+O(\hbar^{1/3}) 
\eea
which immediately gives the correct scaling law of $\Omega \sim \hbar^{-4/3}$. For the full $\Omega_{ij}$, one can use the contour integral expression \eqref{eq:Pifin} for $\Pi_0$, then using the definition for the weighted Airy function. It is the same calculation as \eqref{1st-der} and \eqref{2nd-der}. 
\end{remark}

\subsection{\label{CORPROOF}$L^2$ mass near the caustic: Proof of Corollary \ref{T:Caustic Mass} }

 \bpf[Proof of Corollary  \ref{T:Caustic Mass}]
  %It suffices to  integrate the result \eqref{pi-tube}. 
A random $L^2-$normalized eigenfunction $\Psi_{\hbar, E}$ is distributed according to the uniform measure on the unit sphere from $V_{\hbar, E}.$ Hence,
\[\Psi_{\hbar, E}\stackrel{d}{=}\sum_{\abs{\beta}=N} \frac{a_\beta}{\abs{\vec{a}}} \phi_{\beta,\hbar}\]
where $\vec{a}=\lr{a_\beta}$ is a standard Gaussian vector and $\abs{\vec{a}}=\lr{\sum_\beta \abs{a_\beta}^2}^{1/2}$ is its length. Write $\vec{\phi} = (\phi_\beta)_\beta$.  Then
\[ \E |\Psi(x)|^2 = \E \abs{\left\lan \frac{\vec{a}}{|a|}, \vec{\phi}(x)\right\ran }^2 =  | \vec{\phi}(x) |^2 \E\abs{\left \lan \frac{\vec{a}}{|a|}, \frac{\vec{\phi}(x)}{| \vec{\phi}(x) | }\right \ran }^2  =  \frac{\Pi_{\hbar}(x,x)}{\dim V_{\hbar, E}}.\]
We then integrate \eqref{pi-tube-1} of Theorem \ref{SCLintro} over
  $T_{\delta}(\ccal_E) $ and use \eqref{dimV} and the equation for
  weighted Airy functions in the  second line of \eqref{eq:Ai_kb}
  of  Appendix \S \ref{AIRYAPP}  to find that
$$\begin{array}{lll} \E M_2(\Psi_{\hbar, E}, T_{\delta}(\ccal_E))  & = &\frac{1}{\dim V_{\hbar,E}} \int_{T_{\delta}(\ccal)} \Pi_{\hbar, E}(x,x)dx \\&& \\& \approx &
\Gamma(d) (2\hbar)^{d-1} \int_{\ccal} \int_{-\kappa}^\kappa \Pi_{\hbar, E}(x_0(1 + \hbar^{2/3} r), x_0(1 + \hbar^{2/3} r)) \hbar^{2/3}d r dS(x_0)  \\&&\\& = & 
\Gamma(d) (2\hbar)^{d-1} \cdot \frac{2\pi^{d/2}}{\Gamma(d/2)}\hbar^{2/3} \int_{-\kappa}^\kappa \Pi_{\hbar, E}(x_0(1 + \hbar^{2/3} r), x_0(1 + \hbar^{2/3} r)) d r   \\ &&\\
& \approx &
\Gamma(d) (2\hbar)^{d-1} \cdot \frac{2\pi^{d/2}}{\Gamma(d/2)}\hbar^{2/3} \int_{-\kappa}^\kappa 2^{-d+1} \pi^{-d/2}\hbar^{(1-2d)/3} \Ai_{-d/2}(2r) d r \\ \\
&=& \frac{\Gamma(d)}{\Gamma(d/2)} \hbar^{d/3}\int_{-2\kappa}^{2\kappa}  \Ai_{-d/2}(s) d s \\ \\
&=& \frac{\Gamma(d)}{\Gamma(d/2)^2} \hbar^{d/3}\int_{-2\kappa}^{2\kappa}  \int_0^\infty \Ai(s+\rho) \rho^{d/2-1} d\rho d s
  \end{array}$$
which gives the stated result with
$C(d) = 
\frac{\Gamma(d)}{\Gamma(d/2)^2} $.

 \epf

\subsection{\label{LTENSEMBLE} Scaling limit random Wave Ensemble near the Caustic}

The scaled kernel \eqref{Pi0} is the covariance kernel for a limiting (infinite dimensional) ensemble of Gaussian random functions on $\R^d \simeq T_{x_0} \ccal_E \oplus N_{x_0} \ccal_E$, where $N_{x_0}\ccal_E$ is the fiber of the normal bundle to $\ccal_E$ at $x_0.$ This scaled covariance function corresponds to a Hilbert space of functions on $\R^d$ obtained as scaling limits of Hermite eigenfunctions in the eigenspaces $ V_{\hbar, E}$. 

To show this explicitly, consider the eigenfunctions
$$\widehat{H}_{\hbar} \psi_{\hbar, E} = E \psi_{\hbar}. $$
We rescale this equation around $x_0 \in \ccal_E$ using the local dilation operator
$$D_{\hbar}^{x_0} \psi(u) = \psi(x_0 + \hbar^{\alpha} u). $$
The equation above then is equivalent to
$$D_{\hbar}^{x_0} \circ \widehat{H}_{\hbar} \circ(D_{\hbar}^{x_0})^{-1} \circ D_{\hbar}^{x_0} \psi_{\hbar, E} = E D_{\hbar}^{x_0} \psi_{\hbar}. $$
Now,
$$D_{\hbar}^{x_0} \circ\widehat{H}_{\hbar}\circ (D_{\hbar}^{x_0})^{-1}  = -\half \hbar^2 \hbar^{- 2 \alpha} \Delta_u + \half |x_0 + \hbar^{\alpha} u|^2 
= -\half \hbar^{2 - 2 \alpha} \Delta_u + \frac{|x_0|^2}{2} + \hbar^{\alpha}  \langle x_0, u\rangle  + R(\hbar, x_0, u). 
$$
Note that since $|x_0|^2/2 = V(x_0) = E$,  the terms $E D_\hbar^{x_0} \psi_{\hbar}$ cancel.  The eigenvalue equation becomes a harmonic equation
$$
\left( - \hbar^{2 - 2 \alpha} \Delta_u + 2\hbar^{\alpha}  \langle x_0, u\rangle  + R(\hbar, x_0, u)
\right)\psi^{x_0}_{\hbar}(u) = 0. $$

The equation is a small perturbation of the {\it osculating equation}
$$
\left(  -\hbar^{2 - 2 \alpha} \Delta_u + 2\hbar^{\alpha}  \langle x_0, u\rangle 
\right)\psi^{x_0}_{\hbar}(u) = 0. $$
If $\alpha = \frac{2}{3}$ all factors of $\hbar$ may be eliminated from the osculating equation, giving
$$
\left(\Delta_u -2 u_1  \right)\psi^{x_0}_{\hbar}(u) = 0. $$
Here we choose coordinates $(u_1, u')$ so that $\langle x_0, u \rangle = u_1$, i.e. $x_0 = (1, 0, \dots, 0)$.\footnote{Here, as above, we set  $E=1/2$ as explained near Eq \eqref{eq:scaling}.}
The osculating equation is separable, and becomes
\begin{equation}\label{OE} \left((\frac{\partial^2}{\partial u_1^2} - 2 u_1) + \Delta_{u'} \right) \psi^{x_0}  = 0. \end{equation}
We write
$$\psi^{x_0}(u_1, u') = f(u_1) g(u') $$
to get on $\R_{u_1} \times \R^{d-1}_{u'},$
\be (\frac{\partial^2}{\partial u_1^2} - 2 u_1) f(u_1) = \lambda f(u_1), \;\;\; \Delta_{u'} g (u') = - \lambda g(u'). \ee
We define $\hcal_{\infty}$ be the space of temperate solutions of \eqref{OE},
i.e. solutions in $\scal'(\R^d)$.

The temperate eigenfunctions of $ y'' - x y = \lambda y$ (i.e lying in
$\scal'(\R)$)  are $\{\Ai(x + \lambda)\}. $ The spectrum of the Airy operator is purely absolutely continuous with multiplicity one on all of $\R$ (\cite{G,T} and \cite{O} (Chapter 6)).  Note that $y''  - x y = \lambda u$ is equivalent to $y'' - \tilde{x} y = 0$ if $\tilde{x}=x-\lambda$. 
It follows that a basis of temperate solutions of \eqref{OE} are product
solutions
\begin{equation} \label{PRODUCTS} (2\pi)^{(1-d)/2}  \Ai(2^{1/3}(u_1 + p^2/2)) 2^{1/3}e^{i \lan p, u' \ran}, \;\; p \in \R^{d-1}, \;\; \lambda = p^2. \end{equation}

There is a natural isomorphism $W: L^2(\R^{d-1}) \to \hcal_{\infty},$ which can be used to endow $\hcal_{\infty} $ with an inner product.
Taking the Fourier transform of the osculating equation gives
$$ \left(- |\xi'|^2 - |\xi_1|^2 + \frac{2}{i} \frac{\partial}{\partial \xi_1} \right) \widehat{\psi} = 0. $$
This is a first order linear equation with ``time parameter'' $\xi_1$, and
we write $\widehat{\psi} = \widehat{\psi}(\xi_1, \xi'). $ Then the Cauchy problem 
$$\left\{ \begin{array}{l} \frac{1}{i} \frac{\partial}{\partial \xi_1} \widehat{\psi} (\xi_1, \xi') = 
\frac{1}{2}( |\xi'|^2 + |\xi_1|^2 ) \widehat{\psi}(\xi_1, \xi') \\ \\
\psi(0, \xi') = \widehat{\psi}_0(\xi') \end{array} \right.$$
is solved by the unitary propagator $U(\xi_1)$ on $L^2(\R^{d-1}, d \xi')$ defined by 
\begin{equation} \label{FT1} \widehat{\psi}(\xi', \xi_1) = U(\xi_1) \widehat{\psi}_0(\xi') =  e^{\frac{i}{2} (\xi_1 |\xi'|^2 + \xi_1^3/3)}  \widehat{\psi}_0(\xi').\end{equation}
It follows that  $\hcal_{\infty}$ is isomorphic to the space of Cauchy
data $\widehat{\psi}_0$.
\begin{lem} Let  $g \in L^2(\R^{d-1},dx)$ and let $G(\xi_1, \xi'):  =  U(\xi_1) g(\xi')$. Also, let  $\fcal^*$ denote the
inverse Fourier transform on $\R^d$.  Then the  linear isomorphism 

\begin{equation} \label{W}
\begin{array}{lll} W g(u_1, u') = \left[\fcal^*G \right] (u_1, u') & = & \int_{\R^{d-1} \times \R} e^{i \langle u', p  \rangle} e^{ i \xi_1 u_1}   e^{\frac{i}{2} (\xi_1 |p|^2 + \xi_1^3/3)}  g(p) d\xi _1 dp \\&&\\
& = &\int_{\R^{d-1}} e^{i \langle u', p  \rangle}  \Ai(2^{1/3}(u_1 + |p|^2/2))2^{1/3}
%e^{ i \xi_1 u_1}   e^{i (\xi_1 |p|^2 + \xi_1^3/3)} 
 g(p)  dp. \end{array}  \end{equation} 
 maps $L^2(\R^{d-1}) \to \hcal_{\infty}$ bijectively.
 \end{lem}
 
\begin{proof} $W$ is obviously injective and takes its values in $\hcal_{\infty}$. To prove surjectivity, we let
 $$g_{v_1, v'}( p) = (2\pi)^{-(d-1)/2} e^{ -i \langle v', p \rangle}   \Ai(2^{1/3}(v_1 + |p|^2/2))2^{1/3} \in L^2(\R^{d-1}, dx). $$
 Explicitly,
  $$||g_{v_1, v'}||^2_{  L^2(\R^{d-1})} = C_d \int_0^{\infty} |\Ai(2^{1/3}(v_1 + \rho))|^2 \rho^{\frac{d-1}{2}} d \rho. $$
  for some constant $C_d$. 
 The integral converges since $|\Ai(2^{1/3}(v_1 + \rho))|^2$ decays exponentially
 as $\rho \to \infty$ in $\R_+$.
 We then observe that
   \begin{equation} \label{W2}W g_{v_1, v'}(u_1, u') = \Pi_0(u_1, u', v_1, v'). \end{equation}
But  as observed above, product solutions \eqref{PRODUCTS} span
$\hcal_{\infty}$ and we obtain all of them in \eqref{W}.

The inverse $W^{-1}$ can be explicitly described as follows: 
the range of the Fourier transform restricted to $\hcal_{\infty}$,
$$\fcal: \hcal_{\infty} \to \scal'(\R^d)$$
 is the subspace of temperate functions satisfying the functional
equation \eqref{FT1}. Thus 
$$ W^{-1}: \psi \in \hcal_{\infty} \to \widehat{\psi}(0, \xi')   $$
is an injective map to $L^2(\R^d)$ which inverts $W$.
We may write  $W^{-1} = \rcal \fcal$ where $\rcal F(\xi') = F(0, \xi')$.
\end{proof}

 \begin{defn} \label{WDEF} We define  an inner product on $\hcal_{\infty}$  by
 $$\langle W  g, W h \rangle_{\hcal_{\infty}}:  =\langle g, h \rangle_{L^2(\R^{d-1})}. $$
 \end{defn}
  
 \begin{lem}\label{PROJ} With the above inner product, $\Pi_0^2 = \Pi_0$. \end{lem}
 
 \begin{proof} By \eqref{W},
 $$\begin{array}{lll}\Pi_0^2(u_1, u', v_1, v') & = &  \langle W g_{u_1,u'}, W g_{v_1, v'} \rangle \\&&\\ &  = & \langle g_{u_1, u'}, g_{v_1, v'} \rangle_{L^2(\R^{d-1}} \\&&\\
 & = & \Pi_0(u_1, u', v_1, v'). \end{array}$$
 
 \end{proof}

We also could use $\Pi_0$ directly to define an inner product on $\hcal_{\infty}$
using the method of reproducing kernel Hilbert spaces. According to the Aronszajn theorem, a  symmetric and positive definite
kernel defines a unique reproducing kernel Hilbert space (RKHS) \cite{A}. We
briefly recall that 
a kernel $K(x,y)$ on a space $X \times X$  is called positive-definite if
 %$K(y, x) = \overline{K(x, y)}$ and if 
%$$\sum_{j, k = 1}^M K(x_j, j_k) c_j \bar{c_k} > 0 , \;\;\; \forall \{x_1, \dots, %x_M \in X\}, \forall c_1 \dots c_n \in \C.$$
% That is,
$\begin{pmatrix}K(x_i, x_j)\end{pmatrix}_{i, j \leq M}$ defines a positive Hermitian matrix. Since $\Pi_0$ is the limit of positive definite kernels, it is positive definite
and therefore induces an inner product on $\hcal_{\infty}$. We claim that the RKHS is the same $\hcal_{\infty}$ equipped with  the inner product of Definition \ref{WDEF}.

  By definition, the RKHS associated to $\Pi_0$ is the closure of the 
  set of functions of the form
$$ \left \{g_{v_1, v', \vec a}(\cdot) = \sum_{j = 1}^n a_j \Pi_0(\cdot,  v_1, v'),  \;\; a_j \in \R\right\}, $$ 
equipped  with the inner product
$$\langle g_{v_1, v', \vec a}, g_{u_1, u', \vec b}  \rangle_H = \sum_{j, k= 1}^n a_j b_k  \Pi_0(u_1, u', v_1, v'), $$
and it follows by \eqref{W2} and Lemma \ref{PROJ} that this inner product is the same as Definition \ref{WDEF}.
%where
%$$ f(x) = \sum_{j = 1}^n a_j K(x_j, \cdot), \;\;\; g(x) =  \sum_{j = 1}^m b_j %K(y_j, \cdot). $$
%Since $K$ is positive definite,  this is a Hilbert space inner product, and
%$$\langle K(x, \cdot), K(y, \cdot) \rangle_H : = K(x, y). $$
%It follows that for any $g$,
%$$\langle g, K(x, \cdot) \rangle_H = g(x). $$

The scaled density of the random nodal set in Theorem \ref{CAUSTIC}
can be identified as the density of zeros of the Gaussian random functions in
$\hcal_{\infty}$.

\section{Completion of the proof of Theorem \ref{CAUSTIC}}

The formulae  \eqref{omega-tube} give scaling
asymptotics for the entries of the Kac-Rice matrix of Lemma \ref{L:Gaussian KR}. As mentioned in the remark after the statement of Theorem \ref{CAUSTIC}, we still need to prove the positivity of the first term in our expansion of the spectral projector near the caustic in Theorem \ref{SCLintro}. The proof is supplied by
\bp
\label{Ai-pos}
 $\Ai_{-d/2}(s) > 0$ for all integers $d \geq 2$ and $s \in \R$.
\ep
\bpf
By Proposition \ref{ppBPD},  as $s \to \infty$, $\Ai_k(s) \to 0$. Also note that
 \[ \frac{d}{ds} \Ai_k(s) = -\Ai_{k+1} (s) \]
we know that
\[ \Ai_k(s) = \int_s^\infty \Ai_{k+1}(s') ds'.\]
Hence, it suffices to show that $\Ai_{-1}(s) >0 $ and $\Ai_{-3/2}(s) > 0$ for all $s \in \R$. 

For $\Ai_{-3/2}(s)$, we use Lemma \ref{airy-double} and set $x=y$ to get
\[ \Ai_{-1/2}(s) = \sqrt{2\pi} 2^{1/6} \Ai^2(2^{-2/3} s) \geq 0 \]
and integrate to get $\Ai_{-3/2}(v) > 0$ for all $v \in \R$.

For $\Ai_{-1}(s)$, we use the fact that
$ \Ai_0(s) = \Ai(s)$
and the first zero of $\Ai(s)$ is at $s \approx -2.8$ to see  $\Ai_{-1}(s)>0, \forall s \geq -2$. We now show that $\Ai_{-1}(s) > c > 0$ when $s < -2$ as well.  We use the method of stationary phase, following \cite{BPD}.  We first deform the contour integral $\ccal$ into union of $\ccal_1, \ccal_2, \ccal_3$, see Figure \ref{fig:airy-contour}. 
% has an explicit upper bound on their modulus. 
Here and below we use the same notation for a contour $\ccal_j$ and the integral over the contour. 

The contribution from $\ccal_1$ to $\Ai_{-1}(s)$ for $s < -2$ is
\[ \ccal_1 =  \frac{1}{2 \pi i }\int_{\ccal_1} e^{\frac{T^3}{3} +T |s|} \frac{d T}{T}   =  \frac{1}{2 \pi i }\oint_{|T|=\epsilon} e^{\frac{T^3}{3} +T |s|} \frac{d T}{T}  = 1.\]
The contribution from $\ccal_2$ and $\ccal_3$ are complex conjugates of each other, hence it suffices to compute the real part of one of them. Here we deform the contour $\ccal_2$ again for computing the upper bound.  The sum of the contributions from $\ccal_2$ and $\ccal_3$ is

\bea  
\ccal_2+\ccal_3& = & 
2 \Re \left \{ \frac{1}{2 \pi i }\int_{\ccal_2} e^{\frac{T^3}{3} +T |s|} \frac{d T}{T} \right \} \\
& = & 2 \Re \left \{ \frac{1}{2 \pi i } \left( \int_{T=-\infty}^{-\sqrt{|s|} } + \int_{T= \sqrt{|s|} e^{i \theta},\theta \in (\pi, \pi/2) } + \int_{T = \sqrt{|s|} (i + \exp( i \pi/4) \rho)}  \right) e^{\frac{T^3}{3} +T |s|} \frac{d T}{T} \right \} \\
& = & 2 \Re \left \{ \frac{-1}{2 \pi i } \int_{\sqrt{|s|}}^\infty e^{-\rho^3/3 - \rho |s|} \frac{d\rho}{\rho} \right \} + 2 \Re \left \{ \frac{1}{2\pi} \int_{\pi}^{\pi/2} e^{|s|^{3/2}(e^{3 i \theta}/3 + e^{i\theta})} d\theta \right \} \\
&& \ + 2 \Re \left \{ \frac{1}{2\pi i} \int_0^\infty e^{- |s|^{3/2} \rho^2 -\frac{|s|^{3/2}}{3\sqrt{2}} \rho^3 } e^{i (2/3) |s|^{3/2} + i / (3\sqrt{2}) |s|^{3/2} \rho^3} \frac{d \rho}{e^{i \pi /4} + \rho}  \right \} \\
&=: & K_1 + K_2 + K_3.
\eea
$K_1$ is the real part of a purely imaginary number, hence 
$ K_1 =  0. $ 
For $K_2$, we bound the integrands by the sup-norm hence
\[ K_2 = 2 \Re \frac{1}{2\pi} \int_{\pi}^{\pi/2} e^{|s|^{3/2}(e^{3 i \theta}/3 + e^{i\theta})} d\theta \geq  - \frac{1}{ \pi} \int_{\pi}^{\pi/2} \left| e^{|s|^{3/2}(e^{3 i \theta}/3 + e^{i\theta})} \right| d\theta \geq - \frac{1}{\pi} \int_{\pi}^{\pi/2} d\theta = -0.5,  \]
where we used the fact that $\cos(3\theta)/3 + \cos(\theta) \leq 0$ for $\theta \in [\pi/2, \pi]$. And for $K_3$ we have
\bea 
K_3 &\geq& - \frac{1}{ \pi} \left| \int_0^\infty e^{- |s|^{3/2} \rho^2 -\frac{|s|^{3/2}}{3\sqrt{2}} \rho^3 } e^{i \frac{2}{3} |s|^{3/2} + i |s|^{3/2} \rho^3/(3\sqrt{2})} \frac{d \rho}{e^{i \pi /4} + \rho} \right|\\
& \geq & -\frac{1}{\pi} \int_0^\infty e^{- |s|^{3/2} \rho^2} d\rho  =  -\frac{1}{2\sqrt{\pi} |s|^{3/4}} \geq -\frac{1}{2\sqrt{\pi} 2^{3/4}} \geq -0.2
\eea
Thus, we get $\ccal_1+ \ccal_2 + \ccal_3  \geq 0.3$ when $s < -2$, hence $\Ai_{-1}(s) > 0$ for all $s \in \R$. 

\epf

\section{Intersections of the nodal set with the caustic: Proof of Theorem \ref{T:Caustic Zeros} }

Theorem \ref{T:Caustic Zeros} follows from our formulae for \eqref{COV} restricted to the caustic together with the   Kac-Rice formula for the expected number of intersections of the nodal set with
the caustic. The result is analogous to the formula in \cite{TW} (Proposition 3.2)
for the expected number of intersections of nodal lines with the boundary
of a plane domain. The argument that the Kac-Rice formula can be applied to measure the volume of nodal intersections with the caustic is identical to that given in the beginning of \S \ref{S:KRSECT}.

\bpf

In the notation of the previous section,  with  $x=x_0 + \hbar^{2/3} u$, with $x_0^2 = 1$,  $x_0 = (1,0,\cdots, 0)$ and  $s=2 \langle x_0, u\rangle$, the restriction to the caustic is $u = 0$ and therefore $s = 0$. The covariance matrix used in the Kac-Rice formula is
\begin{align}
\notag \Omega_{ij}(x_0) &= \lr{\dell_{u_i}\dell_{v_j}|_{u=v}\log \Pi_{\hbar, E}(x_0+u\hbar^{2/3},\, x_0+v\hbar^{2/3})}_{2\leq i,j\leq d}\\
&= \hbar^{-4/3} \frac{ \delta_{ij} \Ai_{-1-d/2}(0) \Ai_{-d/2}(0)}{2\Ai^2_{-d/2}(0)} (1+O(\hbar^{1/3})) = \hbar^{-4/3} \frac{ \delta_{ij} \Ai_{-1-d/2}(0) }{2\Ai_{-d/2}(0)} (1+O(\hbar^{1/3})). \label{OMEGA}
\end{align}
We calculate the constant $\fcal_{C_E, d}$ as follows: 
\bea
\fcal_{C_E,d} &=& (2\pi)^{-d/2} \int_{\R^{d-1}} |\xi| \left(\frac{\Ai_{-1-d/2}(0)}{2\Ai_{-d/2}(0)}\right)^{1/2} e^{-\xi^2/2} d\xi \\
&=& (2\pi)^{-1/2} \left(\frac{\Ai_{-1-d/2}(0)}{2\Ai_{-d/2}(0)}\right)^{1/2}  \int_{\R^{d-1}}  |\xi|  e^{-\xi^2/2} \frac{d\xi}{(2\pi)^{(d-1)/2}}  \\
%&=& (2\pi)^{-1/2} \left(\frac{\Ai_{-1-d/2}(0)}{2\Ai_{-d/2}(0)}\right)^{1/2}  %\frac{\int_{\R^{d-1}}  |\xi|  e^{-\xi^2/2} \frac{d\xi}{(2\pi)^{(d-1)/2}}}%{\int_{\R^{d-1}}   e^{-\xi^2/2} \frac{d\xi}{(2\pi)^{(d-1)/2}}} \\
%&=& (2\pi)^{-1/2} \left(\frac{\Ai_{-1-d/2}(0)}{2\Ai_{-d/2}(0)}\right)^{1/2}  %\frac{\int_0^\infty e^{-r^2/2} r^d dr }{\int_0^\infty e^{-r^2/2} r^{d-1} dr } \\
%(r=\sqrt{2u}) &=& (2\pi)^{-1/2} \left(\frac{\Ai_{-1-d/2}(0)}{2\Ai_{-d/2}%(0)}\right)^{1/2}  \sqrt{2} \frac{\int_0^\infty e^{-u} u^{(d+1)/2-1} du }%{\int_0^\infty e^{-r^2/2} u^{d/2-1} du } \\
&=&  (2\pi)^{-1/2} \left(\frac{\Ai_{-1-d/2}(0)}{\Ai_{-d/2}(0)}\right)^{1/2}  \frac{\Gamma(d/2)}{\Gamma((d-1)/2)}
\eea
In the case of $d=2$, we have
\[ \fcal_{C_E,2} = \frac{1}{\sqrt{2} \pi} \left(\frac{\Ai_{-2}(0)}{\Ai_{-1}(0)}\right)^{1/2},\]
and hence 
\[ \E \left( \# Z_{\Phi_{\hbar, E}} \cap \ccal_E \right) = \hbar^{-2/3}  \fcal_{C_E,2} (2\pi) = \hbar^{-2/3} \sqrt{2} \left(\frac{\Ai_{-2}(0)}{\Ai_{-1}(0)}\right)^{1/2}. \]
This concludes the proof.
\end{proof}

\section{Allowed and forbidden annuli for $\alpha < \frac{2}{3}$: Proof of Theorem \ref{T:Mainb}\label{ANNULISECT}}
The main results of this section are Propositions \ref{pp:allowed} and \ref{pp:forbidden}, giving asymptotic
formulae for the Kac-Rice matrix in $h^{\alpha}$ tubes around the caustic where $0 < \alpha < 2/3. $ 
In Section \ref{S:Allowed Annuli} we find the asymptotics in the allowed region for $\alpha < 2/3. $ In Section \ref{pp:forbidden} we do
the same in the forbidden region.

\begin{remark}
Before going into the details of the proofs, we explain how the asymptotics
for $0< \alpha < \frac{2}{3}$ are related to those for $\alpha=0$ or $\alpha=2/3$. The asymptotics have a leading term and a remainder term.  The leading term for  $0< \alpha < \frac{2}{3}$ can be formally obtained by interpolation from the leading term of the  $\alpha=0$ result or the $\alpha=2/3$ result. But this would  not prove that the asymptotics are valid,
because one still has to prove that
 the remainder term for $0< \alpha < \frac{2}{3}$  is smaller than the
purported leading term. The proof we give uses the  stationary phase approach similarly to $\alpha = 0$ but  with a nearly degenerate quadratic phase function and keeps track of how the degeneracy affects the remainder estimate.
\end{remark}

\subsection{The allowed annuli region}\label{S:Allowed Annuli}
\begin{prop}\label{pp:allowed}
Let $\abs{x}^2 = 1 - \hbar^{\alpha} s$ for any $s>0, \alpha \in (0, 2/3)$. We have
\[ \Pi_{\hbar}(x,x) = C_d\hbar^{\Delta(\alpha,d)} s^{d/2-1} (1 + O(\hbar^{1- \frac{3}{2}\alpha})+O(\hbar^{\alpha/2})),\]
where 
\[\Delta(\alpha,d)=1-d + \frac{d-2}{2}\alpha\qquad \text{and}\qquad C_d=\frac{2^{-d+1}\pi^{-d/2}}{\Gamma\lr{\frac{d}{2}}}.\] 
Further, 
\[  \Omega_{ij} (x)=d^{-1}\delta_{ij} \hbar^{-2+\alpha} s (1 + O(\hbar^{1-(3/2)\alpha}) + O(\hbar^{\alpha/2})).\]
The implied constants in $O(\cdots)$ are uniform for $s$ in a compact subset of $(0,\infty).$
\end{prop}
\bpf
We give the details for the calculation for $\Pi_{\hbar}(x,x).$ The outline of the proof is given in Section \ref{S:Outline} and our starting point is \eqref{eq:Pi_Contour}. As explained in Section \ref{S:Outline}, the integrand $A(z)e^{\Phi(z)/\hbar}$ is defined on 
\[S=\C \backslash \lr{(\infty, -1]\cup [1,\infty)\cup \set{0}}.\] 
Since the integrand is holomorphic, we may deform the contour $C_\epsilon=\set{\abs{z}=1-\epsilon}$ within $S$. For the portion near $-1$, we break $C_\epsilon$ into two pieces, one terminating at $-1+\epsilon+i0$ and the other starting at $-1+\epsilon -i0.$ Moreover, since the integrand in \eqref{eq:Pi_Contour} is rapidly decaying at $\infty \pm i\tau$ for every $\tau>0$, we may deform $C_\epsilon$ all way to $\infty \pm i0 $ inside $S$ (as shown in Figure \ref{fig:safe-region}). Finally, since $z_{\pm}\in S,$ we can make the deformed contour pass through $z_{\pm}$, and in a $\hbar^{\alpha/2}$ 
neighborhood be given by 
\begin{equation}
z_\pm+\hbar^{\alpha/2}\eta\lr{1\mp \frac{i}{\sqrt{3}}},\qquad \eta \in \R.\label{E:critcont}
\end{equation}
This completes step (i) of the proof outline from Section \ref{S:Outline}. We have 
$$\int_{|z| = 1 - \epsilon} A(z)e^{\Phi(z)/\hbar}dz= \sum_{\pm} \underbrace{\int_{C_{\pm}} A(z)e^{\Phi(z)/\hbar}dz}_{I_\pm}+ \underbrace{\int_{C_0}A(z)e^{\Phi(z)/\hbar}dz}_{I_0},$$
where $C_0$ is the keyhole contour that starts at $\infty+i0,$ wraps around $1$ and ends at $\infty-i0$ and $C_{\pm}$ are the complex conjugate contours on which $\pm \Im z>0$ shown in the left panel of Figure \ref{fig:safe-region}. The purpose of requiring \eqref{E:critcont} is that 
\[\Re{\Phi\lr{z_++\hbar^{\alpha/2}\eta\left[1-\frac{i}{\sqrt{3}}\right]}}=-\hbar^{3\alpha/2}\frac{\eta^2\sqrt{s}}{2\sqrt{3}}+O(\hbar^{2\alpha}).\]
Therefore, there exists $K>0$ and $C>0$ so that for all $s$ in a compact subset of $\R_+$
\begin{equation}
\sup_{z\in {C_0, C_{\pm}\\ \abs{z-z_{\pm}}}\geq K\hbar^{\alpha/2}} \Re \Phi(z) < - C \hbar^{3\alpha/2}.\label{E:PhaseEst}
\end{equation}
This ensures that $\exp\lr{\Phi(z)/\hbar}=O(\hbar^\infty)$ uniformly for $z\in  C_0, C_{\pm}$,  satisfying $\abs{z-z_{\pm}}\geq K\hbar^{\alpha/2}$. It will turn out that $I_0$ grows more rapidly as $\hbar \rightarrow 0$ than $I_{\pm}$. The exact growth rate in $\hbar$ of $I_{\pm}$ is given in the following Lemma. 
\begin{lem}\label{LEMOMEGAPROP}
There exists a constant $C>0$ (depending on $u$) so that as $h\gives 0$ 
  \begin{equation}
\abs{I_\pm}= C\cdot \hbar^{\delta(\alpha,d)}\lr{1+O(\hbar^{1-3\alpha/2})+O(\hbar^{\alpha/2})},\label{E:Iplus}
\end{equation}
where
\[\delta(\alpha, d)=\frac{1}{2}\lr{1-\frac{\alpha}{2}-d\lr{1+\frac{\alpha}{2}}}.\]
\end{lem}
\begin{proof}
It suffices to consider $I_+$ since $I_+=\overline{I_-}.$ We have by \eqref{CPE},
\[z_+ = -1+ 2 |x|^2 + i 2 |x| \sqrt{1-|x|^2} = 1 + i \cdot 2\hbar^{\alpha/2} \sqrt{s} + O(\hbar^\alpha)\]
and $\Re{\Phi(z_+)}=0$ since $\Re(\Phi)$ vanishes on the entire unit circle (except at the point $z=-1$). Recall the constant $K$ from \eqref{E:PhaseEst}. By \eqref{E:PhaseEst}, we have $I_+=I_+'+O(\hbar^\infty),$  where
\[I_+'=\int_{C_+}e^{\Phi(z)/\hbar}A(z) \chi_K(z)dz\]
and
\[\chi_K(z)=
\begin{cases}
  1, & \abs{z-z_+}\leq \frac{K}{2}\hbar^{\alpha/2}\\ 0, &\abs{z-z_+}>K \hbar^{\alpha/2}
\end{cases}.
\]
To evaluate the localized integral, we seek to apply the method of stationary phase (Proposition \ref{P:MSP}) to $I_+,$ and we need the following Lemma to ensure that the error terms are uniformly bounded. 
\begin{lem}\label{L:SPEst}
There exists a $K>0$ and $c, C_j>0,\,\, j=1,2,3,4,$ so that 
\begin{equation}
\inf_{z\in \text{supp}(\chi_K)\cap C_+}\abs{\Phi''(z)}>c\hbar^{\alpha/2}.\label{E:HessLB}
\end{equation}
Futher, for $j=1,2,3$ we have
\begin{equation}\label{E:SPerror1}
 \sup_{z\in \text{supp}(\chi_K)\cap C_+}\abs{\dell_z^j \Phi(z)}\leq C_j \hbar^{(3-j)\frac{\alpha}{2}}
\end{equation}
while for $j=4$ 
\begin{equation}\label{E:SPerror2}
 \sup_{z\in \text{supp}(\chi_K)\cap C_+}\abs{\dell_z^j \Phi(z)}\leq C_j .
\end{equation}
\end{lem}
\begin{proof}
The estimates \eqref{E:HessLB},\eqref{E:SPerror1}, and \eqref{E:SPerror2} are all obtained by Taylor expansion using that
\[\Phi'(z_+)=0\quad \text{and}\quad \Phi''(z_+)=2i\hbar^{\alpha/2}\sqrt{u}\lr{1+O(\hbar^{\alpha/2})}\]
as well as $\sup_{z\in \text{supp}(\chi_K)}\abs{\Phi^{(m)}(z)}\leq \kappa$ for some $\kappa$ independent of $K$ and $m=3,4.$ 
\end{proof}
We change variables 
\[z=z_++\hbar^{\alpha/2} \eta,\]
and $\twiddle{I}_+'$ becomes  
\[\hbar^{\alpha/2}e^{\Phi(z_+)/h}\cdot \int_{\twiddle{C}_+}\twiddle{\psi}(s)\twiddle{\chi}_K(\eta)\twiddle{A}(\eta)\exp\left[\hbar^{-1+3\alpha/2}\cdot \twiddle{\Phi}(\eta)\right]d\eta,\]
where $\twiddle{\chi}_K(\eta)=\chi_K(z_+\hbar^{\alpha/2}\eta)$ and we have written
\[\twiddle{\Phi}(\eta):=\hbar^{-3\alpha/2}\lr{\Phi(z_++\hbar^{\alpha/2}\eta)-\Phi(z_+)}\]
as well as
\[\twiddle{A}(\eta):=A (z_++\hbar^{\alpha/2}\eta).\] 
By \eqref{E:SPerror1} and \eqref{E:SPerror2}, there exists $C_j>0$ so that the phase function $\twiddle{\Phi}$ satisfies
\begin{equation}
\sup_{\eta\in \text{supp}(\twiddle{\chi})}\abs{\dell_\eta^j \twiddle{\Phi}(\eta)}\leq C_j\label{E:PhaseEst1}
\end{equation}
for all $h$ sufficiently close to $0.$ Moreover by \eqref{E:HessLB}, there exists $c>0$ so that
\begin{equation}
\inf_{z\in \text{supp}(\chi)}\abs{\dell_s^2 \twiddle{\Phi}(\eta)}\geq c\label{E:Hess Bound}
\end{equation}
for all $h$ sufficiently close to $0.$ This shows that the constant $C$ in the error term from applying stationary phase (Proposition \ref{P:MSP}) is independent of $h.$ We have
\[ \abs{I_+} = \frac{ \hbar^{\alpha/2}\abs{ A(z_+) e^{\Phi(z_+)/\hbar} }}{\sqrt{2\pi (\abs{\twiddle{\Phi}''(0)} / \hbar^{1-3\alpha/2})}} (1 + R_+) +O(\hbar^\infty)\]
To compute the leading order term, we note that
\[ \Phi(z_+) \in i \R, \quad |\Phi''(z_+)| =2 \hbar^{\alpha/2}\sqrt{s} (1+O(\hbar^{\alpha/2})). \]
Thus $|e^{\Phi(z_+)/\hbar}| = 1$. Since $A(z)$ has the pole contribution $(z-1)^{-d/2}$ and $z_+=1+2i\hbar^{\alpha/2}s$ we have 
\[ A(z_+) =C\cdot \hbar^{-(d/2)(1+\alpha/2)}s^{-d/2} (1+O(\hbar^{\alpha/2})).\] 
Moreover, derivatives of the amplitude $\twiddle{\chi}_K$ in the $\eta$ variable are all bounded and for each $j\geq 0$ there exists $C_j>0$ so that
\[\dell_\eta^j \lr{\twiddle{\chi}_K(\eta)\cdot \twiddle{A}(\eta)}\leq C_j \abs{A(z_+)}.\]
Therefore, there exists $C>0$ so that 
\[ \abs{R_+}\leq C \hbar^{1-3\alpha/2}.\]
Combining the above estimate, we have
\[ \abs{I_+} = C \hbar^{\delta(\alpha, d)}\left(1+O(\hbar^{\alpha/2})+O(\hbar^{1-3\alpha/2})\right),\]
completing the proof of Lemma \ref{LEMOMEGAPROP}.
\end{proof}
We study $I_0$ as $\hbar\rightarrow 0$ in the next Lemma.
\begin{lem}\label{L:I0 Int} As $h\gives 0,$
\[I_0=\frac{2^{1-d} \pi^{-d/2} s^{d/2-1}}{\Gamma\lr{\frac{d}{2}}}\cdot \hbar^{\Delta(\alpha,d)}\lr{1+O(\hbar^{2-3\alpha}) + O(\hbar^{1-\alpha})},\]
where $\Delta(\alpha, d)=-\frac{d}{2}+1-\alpha -\frac{d}{2}(1-\alpha)$.
\end{lem}
\begin{proof}
We introduce a new variable $w = (1-z)/(1+z), $ then $z=(1-w)/(1+w)$. We will abuse notations to mean $A(w) := A(z(w)), \Phi(w) := \Phi(z(w))$. Then we get
\[ A(w) \frac{dz }{dw} = \frac{-1}{\pi i}  \frac{(1-w^2)^{d/2-1}}{(4\pi \hbar)^{d/2}(1+w)^2} w^{-d/2}, \quad \Phi(w) = \hbar^\alpha s w + \sum_{j\geq 3 \,\, \text{}odd}w^j/j,\]
 In the new variable $w$,  
the contour $C_0$  starts from $w=-1+i0$, wraps around $[-1,0]$ (which is the image of the branch cuts from the $z-$plane) counterclockwise and returns to $w=-1-i0$. We deform $C_0$ to 
\[C_{0,-}=[-1-i0,\, -\hbar^{1-\alpha}/u-i0],\, \,C_{0,\text{circle}}=\set{\abs{w}=\hbar^{1-\alpha}/u}\,\, C_{0,+}=[-\hbar^{1-\alpha}/u+i0,\, -1+i0]\]
and correspondingly write 
\[I_0=I_{0,-}+I_{0,\text{circle}}+I_{0,+}.\]
For any $\epsilon>0,$ if $w<-\hbar^{1-\alpha-\epsilon},$ then the phase $\Phi(w)$ is exponentially large in $h$ and negative so that 
\[I_0=\int_{C_0\cap \abs{w}>\hbar^{1-\alpha-\epsilon}} A(w)e^{\Phi (w)/h}dw = O(\hbar^\infty).\]
Let us denote by  $\twiddle{C}_{0,\hbar}$  a contour running from $-\hbar^{-\epsilon} - i0$ counterclockwise around $(-\infty, 0]$ to $-\hbar^{-\epsilon} + i0.$ Changing variables $W=sw\hbar^{\alpha-1}$ and Taylor expanding, we have  \begin{align*}I_0&= 2^{1-d}\pi^{-d/2} \hbar^{\Delta(\alpha,d)}s^{d/2-1}\cdot \frac{i}{2\pi} \int_{\twiddle{C}_{0,\hbar}} (W^{-d/2} + B_\hbar(W)) e^{W  + A_\hbar(W)} dW \\
&= 2^{1-d}\pi^{-d/2} \hbar^{\Delta(\alpha,d)}s^{d/2-1}\cdot \frac{i}{2\pi}\int_{\twiddle{C}_{0}} e^{W } W^{-d/2}dW\,\, \cdot \lr{1+O(\hbar^{2-3\alpha}) + O(\hbar^{1-\alpha})}.
\end{align*}
where in the last step we used the Trimming Lemma \ref{lm:trim}, and $\twiddle{C}_0,$ is a contour running from $-\infty - i0$ counterclockwise around $(-\infty, 0]$ to $-\infty + i0$. To see that the lemma is applicable, we note that
\[A_\hbar(W) =\hbar^{-1} \sum_{j\geq 3 \,\, \text{}odd} (W \hbar^{1-\alpha} s^{-1})^j/j \]
Since there exists $w_0$ small enough such that if $|w| < w_0$, we have  $\sum_{k \geq 1} |w|^{2k+1}/(2k+1) < 2 |w|^{3}/3$. And since $|w| < \hbar^{1-\alpha-\epsilon}$, we can always achieve $w < w_0$ by letting $\hbar$ being small enough. Thus for $\hbar$ small enough, for $|W| < \hbar^{-\epsilon}$ we have 
\[\abs{A_\hbar(W)} \leq \hbar^{-1} (2/3) |W \hbar^{1-\alpha} / s|^{3} = (2/3) |s|^{-3} \hbar^{2-3\alpha} |W|^3 \]
Similarly, we have $\abs{B_\hbar(W)} \leq C \hbar^{1-\alpha} $. This shows we can apply Lemma \ref{lm:trim} with $a = \min(2-3\alpha, 1-\alpha).$ Finally, using Hankel's representation for the reciprocal of the Gamma function
\[\frac{1}{\Gamma(z)}=\frac{i}{2\pi }\int_{\twiddle{C}_0}W^{-z}e^{W}dW,\]
we find
\[I_0=\frac{2^{1-d} \pi^{-d/2}s^{d/2-1}}{\Gamma\lr{\frac{d}{2}}} \hbar^{\Delta(\alpha,d)}\lr{1+O(\hbar^{2-3\alpha}) + O(\hbar^{1-\alpha})}.\]
This completes the proof of Lemma \ref{L:I0 Int}.
\end{proof}

Comparing the contribution from the critical point $z_\pm$ and the pole singularity at $z=1$, 
\be I_\pm / I_0 = C \hbar^{\frac{d-1}{2}(1-3\alpha/2)} (1 + O(\hbar^{1-(3/2)\alpha}) + O(\hbar^{\alpha/2})) \label{eq:I0compare}\ee
we see for $\alpha < 2/3$, $I_0$ dominates. This concludes the $\Pi_{\hbar}(x,x)$ estimates claimed in Proposition \ref{pp:allowed}.

The estimate for $\pa_{x_i} \Pi_{\hbar}(x,x)$ and $\pa_{x_i} \pa_{y_j} \Pi_{\hbar}(x,x)$ are similar, the singularity at $z = 1$ dominates, and each additional $(1-z)$ factor contribute an $\hbar^{1-\alpha}/s$ factor. We get
\bea
\Pi^{-1}_{\hbar} \pa_{x_i} \Pi_{\hbar}(x,x) & = & -\lr{\frac{d-2}{2} }x_i \hbar^{-\alpha}  s^{-1}(1 + O(\hbar^{1-(3/2)\alpha}) + O(\hbar^{\alpha/2})) \\
\Pi^{-1}_{\hbar} \pa_{x_i} \pa_{y_j} \Pi_{\hbar}(x,x) & = & \lr{\lr{\frac{d}{2}-1}\lr{\frac{d}{2}-2} x_i x_j s^{-2} \hbar^{-2\alpha}  +d^{-1}\delta_{ij} \hbar^{-2+\alpha} s} (1 + O(\hbar^{1-(3/2)\alpha}) + O(\hbar^{\alpha/2})) \\
\Omega_{ij} & = & \lr{\frac{2-d}{2} x_i x_j s^{-2} \hbar^{-2\alpha}+d^{-1}\delta_{ij} \hbar^{-2+\alpha} s} (1 + O(\hbar^{1-(3/2)\alpha}) + O(\hbar^{\alpha/2})) 
\eea
As long as $\alpha<2/3,$ the $\delta_{ij}$ term from $\Omega_{ij}$ dominates and we have
\[\Omega_{ij}=d^{-1}\delta_{ij} \hbar^{-2+\alpha} s  (1 + O(\hbar^{1-(3/2)\alpha}) + O(\hbar^{\alpha/2})).\]
This completes the proof of Proposition \ref{pp:allowed}.
\epf
Applying the Kac-Rice formula, we complete the proof of the allowed region part of the Theorem \ref{T:Mainb}.

\subsection{The forbidden annuli region}

\bp
\label{pp:forbidden}
Let $x^2 = 1 + \hbar^{\alpha} s$ for any $s>0, \alpha \in (0, 2/3)$, then 
\[ \Pi_{\hbar}(x,x) = 2^{-d}\pi^{-\frac{d+1}{2}}\hbar^{-\frac{d-1}{2} - \alpha\frac{d+1}{4}} s^{-1/4-d/4} e^{ -2/3 s^{3/2} \hbar^{-1+ 3 \alpha/2}} (1+O(\hbar^{\alpha/2}) +O( \hbar^{1-(3/2)\alpha})).\]
Moreover,
\[  \Omega_{ij} (x) = \hbar^{-1-\alpha/2} (-x_i x_j+ \delta_{ij}) (2 s^{1/2})^{-1} (1 + O(\hbar^{1-(3/2)\alpha}) + O(\hbar^{\alpha/2})).\]
The implied constants in the estimate above are uniform for $s$ in a compact subset of $(0,\infty).$
\ep
\bpf
As in \S \ref{S:Allowed Annuli}, the initial contour $C_\epsilon =\set{\abs{z}^2=1-\epsilon}$ can be deformed freely inside $S$ (defined in \eqref{E:SDef}). The relevant critical point of $\Phi$ is 
\[ z_c = 1 + 2 (x^2-1) - 2 \sqrt{(x^2-1) x^2} = 1 - 2 \sqrt{\hbar^{\alpha} s} + O(\hbar^{\alpha}), \]
with 
\[ \Phi(z_c) = -\frac{2}{3} s^{3/2} \hbar^{3\alpha/2 } + O(\hbar^{5 \alpha/2}), \quad \Phi''(z_c) = \frac{\hbar^{\alpha/2} \sqrt{s} }{2} + O(\hbar^\alpha).\]
We may therefore deform $C_\epsilon$ in a $\hbar^{\alpha/2}$ neighborhood of $z_c$ so that it is parametrized by 
\[z_c+\hbar^{\alpha/2}i\eta,\qquad \eta\in \R.\]
Then a similar stationary phase computation as in the Proposition \ref{pp:allowed} carries through to get the result for $\Pi_\hbar(x,x)$. 

Next, we evaluate $\Omega_{ij}(x)$ using (\ref{eq:omega_int}). Here we write $\Omega$ as
\begin{equation}
\label{eq:omega_int}
\Omega_{ij}(x) = \big\lan \frac{x_i x_j}{\hbar^2} \frac{2(z-w)^2}{(1+z)^2(1+w)^2}  +  \frac{\delta_{ij}}{\hbar} \frac{2z}{1-z^2}  \big\ran_{z,w} 
\end{equation}
where 
\[ \lan \cdots \ran_{z,w} :=  \frac{ \oint_z \oint_w (\cdots) A(z)  A(w)  e^{\Phi(z)/\hbar}e^{\Phi(w)/\hbar}dz dw}{ \oint_z \oint_w   A(z)  A(w)  e^{\Phi(z)/\hbar}e^{\Phi(w)/\hbar}dz dw}.
\]
Applying the standard stationary phase method, we get
\[ \lan f(z,w) \ran_{z,w} =f(z_c,w_c) - \hbar^{1-\alpha/2} s^{-1/2}(f_{zz}+ f_{ww})|_{(z_c,w_c)}( 1 + O(\hbar^{1-\alpha/2}) +O(\hbar^\alpha))  \]
Indeed, in computing the next order term in $\lan f(z,w)\ran_{z,w}$, one need not worry about the pairing between $A(z)$ and $A(w)$ among themselves, since the same contribution from the denominator will cancel them out. Thus
\bea
\lan \frac{x_i x_j}{\hbar^2} \frac{2(z-w)^2}{(1+z)^2(1+w)^2}  \ran_{z,w} &=& 0 +  \left(\frac{x_i x_j}{\hbar^2} \frac{ - \hbar^{1-\alpha/2} s^{-1/2}}{2} \right)(1+ O(\hbar^{1-\alpha/2}) + O(\hbar^\alpha)) \\
\lan  \frac{\delta_{ij}}{\hbar} \frac{2z}{1-z^2}  \ran_{z,w}  &=& \frac{\delta_{ij}}{\hbar} \frac{1}{2 \hbar^{\alpha/2} \sqrt{s}} + O(\hbar^{\alpha/2}) + O(\hbar^{1-(3/2) \alpha}) \\
\Omega_{ij}(u) &=& \lan \frac{x_i x_j}{\hbar^2} \frac{2(z-w)^2}{(1+z)^2(1+w)^2}  +  \frac{\delta_{ij}}{\hbar} \frac{2z}{1-z^2}  \ran_{z,w} \\
&=& \hbar^{-1-\alpha/2} (-x_i x_j + \delta_{ij}) (2 \sqrt{s})^{-1} (1 + O(\hbar^{1-\alpha(3/2)}) + O(\hbar^{\alpha/2})) 
\eea
\epf
Applying the Kac-Rice formula, we complete the proof of the forbidden region part of the Theorem \ref{T:Mainb}.

\appendix

\section{\label{AIRYAPP} The Weighted Airy Functions}
For any $k \in \R$, we define the weighted Airy function $\Ai_k(s)$ by
\begin{equation}
 \label{eq:Ai_k2} \Ai_{k}(s) := \int_\ccal T^{k} \exp \left(  \frac{T^3}{3} - T s\right) \frac{dT}{2\pi i},\qquad s \in \R
\end{equation}
The contour $\ccal$ is coming from $e^{-i \theta_-} \infty$ and ending at $e^{i \theta_+} \infty$, for $\theta_- \in [-\pi/2,-\pi/6]$ and $\theta_+\in [\pi/6, \pi/2]$,  and stays within the right half plane ($\Re\,T>0$) (shown as $C$ in Figure \ref{fig:airy-contour}). 
\begin{figure} 
\centering  
\includegraphics[width=0.33\textwidth]{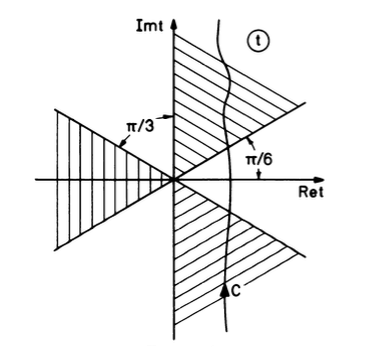}
\includegraphics[width=0.33\textwidth]{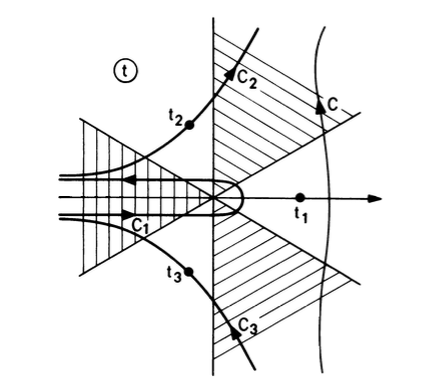}
\caption{Defining contour for the generalized Airy function $\Ai_k(u)$ (left panel), and its deformation when $x\ll0$. (Taken from \cite{BPD}.) \label{fig:airy-contour}}
\end{figure}

If $k=0$, $\Ai_0(s) = \Ai(s)$. For positive integral weight, we have 
\be \Ai_k(s) = (-\frac{d}{dx})^k \Ai(s) \quad k \in \Z_+\ee
For negative real weight, we can use 
\[ T^{-k} = \frac{1}{\Gamma(k)}\int_0^\infty \rho^{k-1} e^{- \rho T} d\rho, \quad \Re(T)>0, k > 0 \]
to get
\bee
\Ai_{-k}(s) & =&   \frac{1}{\Gamma(k)}  \int_\ccal \int_0^\infty\exp \left(  \frac{T^3}{3} - T (s+\rho)\right) \rho^{k-1}  d\rho  \frac{dT}{2\pi i} \notag \\
 &=&  \frac{1}{\Gamma(k)} \int_0^\infty \Ai(s+\rho)\rho^{k-1} d\rho \label{eq:Ai_kb}
\eee

The weighted Airy function has been considered before, see \cite{BPD}, there $W(z,n) = \Gamma(n+1) \Ai_{-(n+1)} (z)$.  Here we quote the weighted Airy functions' asymptotic expansion from the above paper, cf Eq. (17), (18), (19) there. 
\bp[\cite{BPD}]
\label{ppBPD}
Fix any $k \in \R$. \\
(1) For $s \gg 0$, we get 
\[ \Ai_{-k}(s) = \frac{1}{2\sqrt{\pi}} \frac{e^{-2/3 s^{3/2}}}{s^{(2k+1)/4}}[1 - \frac{k^2+2k}{4 s^{3/2}} + O(s^{-3})]. \]
(2) For $s \ll 0$,  we get 
\bea 
\Ai_{-k}(s) &=& \sum_{j=0}^\infty \frac{|s|^{k-3j-1}}{3^j \Gamma(k-3j)} 
 \; + \; \frac{\sin(2 |s|^{2/3}/3 - (2k-1) \pi/4)}{\pi^{1/2} |s|^{(2k+1)/4}} [1+O(|s|^{-1})]
\eea
In first term summation, if $j$ is such that $k - 3j$ is a non-positive integer, then $|\Gamma(k-3j)| = \infty$ and the corresponding term vanishes. 
\ep
The only difference with \cite{BPD} is that for $s \gg 0$ case, we computed the second order term, and we removed the $k \geq 1$ condition. Since the arguments in that paper still holds verbatim, we will not give the detail of the proof here.  
%This asymptotic formula can be used to extrapolate the result from the caustic to the result on allowed and forbidden annuli. 

The following lemma is used in the proof of Theorem  \ref{SCLintro}.
\bl[Product formula for Airy function]
\label{airy-double}
\[ \Ai(x) \Ai(y) =  \int_{\ccal} e^{ \frac{2 T^3}{3} - (x+y) T - \frac{(x-y)^2}{8 T} } \sqrt{\frac{1}{2 \pi T}} \frac{d T}{(2 \pi i)} \]
In particular, if $x=y$, we get 
\[ \Ai(x)^2  = 2^{-1/6} (2\pi)^{-1/2} \Ai_{-1/2}(2^{2/3} x). \]
\el

\bpf
From the integral expression of airy function we obtain
\[\Ai(x) \Ai(y)  =  \int_{T_1 \in \ccal} \int_{T_2 \in \ccal} e^{\frac{T_1^3+T_2^3}{3} -  x T_1-  y T_2}  \frac{d T_1}{2 \pi i} \frac{d T_2}{2 \pi i}.\]
If we straighten $\ccal$ to be $i \R + \epsilon$ for some $\epsilon>0$, then we can reparametrize the integration variables as $T_1 = T+ S/2, T_2 = T- S/2$, for $T \in i \R + \epsilon$ and $S \in i \R$. The integration becomes
\bea
\Ai(x) \Ai(y) &=& \int_{T \in i \R + \epsilon} \int_{S \in i\R} e^{ \frac{2 T^3}{3} - (x+y) T + \frac{1}{2} (S^2 T - S (x-y)) } \frac{dT dS}{(2\pi i)^2}  \\
&=& \int_\ccal e^{ \frac{2 T^3}{3} - (x+y) T - \frac{(x-y)^2}{8 T} } \sqrt{\frac{1}{2 \pi T}} \frac{d T}{(2 \pi i)} 
\eea
\epf

\section{\label{APPB} Analytic Stationary Phase Method}
We recall here the version of the method of stationary phase that we will use. 
\begin{prop}[Thm. 7.7.5 Vol. $3$ \cite{Hor}]\label{P:MSP}
  Let $K\subseteq \R^n$ be a compact set, $X$ and open neighborhood of $K$ and $k$ be a positive integer. If $u\in C_0^{2k}(K),\, f\in C^{3k+1}(X)$ and $Im(f)\geq 0$ in $X$, $Im(f(x_0))=0, f'(x_0)=0,\det f''(x_0)\neq 0$ and $f'\neq 0$ on $K\backslash \{x_0\},$ then
  \begin{align*}
 &   \abs{\int_X u(x)e^{i\omega f(x)}dx-e^{i\omega f(x_0)}\cdot \det\lr{\omega f''(x_0)/2\pi i}^{-1/2}\sum_{j<k }\omega^{-j}L_j u}\\
  &~~  \leq C\omega^{-k}\cdot \det\lr{\omega f''(x_0)/2\pi i}^{-1/2} \sum_{\abs{\alpha}\leq 2k}\sup \abs{D^\alpha u}.
  \end{align*}
The constant $C$ is uniform over any bounded set $S$ in $C^{3k+1}$ as long as $\abs{\det f''(x)}$ is uniformly bounded away from $0$ for all $f\in S.$
\end{prop}

\end{document}